%% file: main.tex
\documentclass[sigconf]{acmart}
\usepackage{graphicx} % Required for inserting images
\usepackage[linesnumbered, noend]{algorithm2e}
\usepackage{float}
\usepackage{caption}
\usepackage{subcaption}
\input{helpers}

\AtBeginDocument{%
  }

\author{Walid A. Hanafy}
\orcid{0000-0001-5765-8194}
\affiliation{%
  \institution{University of Massachusetts Amherst}
  \city{Amherst}
  \state{MA}
  \country{USA}
}
\email{whanafy@cs.umass.edu}

\author{Li Wu}
\orcid{0000-0001-7115-1517}
\affiliation{%
  \institution{University of Massachusetts Amherst}
  \city{Amherst}
  \state{MA}
  \country{USA}
}
\email{liwu@umass.edu}

\author{David Irwin}
\orcid{0000-0003-1722-4927}
\affiliation{%
  \institution{University of Massachusetts Amherst}
  \city{Amherst}
  \state{MA}
  \country{USA}
}
\email{irwin@ecs.umass.edu}

\author{Prashant Shenoy}
\orcid{0000-0002-5435-1901}
\affiliation{%
  \institution{University of Massachusetts Amherst}
  \city{Amherst}
  \state{MA}
  \country{USA}
}
\email{shenoy@cs.umass.edu}

\renewcommand{\shortauthors}{Hanafy et al.}

\copyrightyear{2026}
\acmYear{2026}
\setcopyright{cc}
\setcctype{by}
\acmConference[SoCC '26 Vol. 2]{ACM Symposium on Cloud Computing Volume 2}{November 18--20, 2026}{Singapore, Singapore}
\acmBooktitle{ACM Symposium on Cloud Computing Volume 2 (SoCC '26 Vol. 2), November 18--20, 2026, Singapore, Singapore}
\acmDOI{10.1145/3828161.3857880}
\acmISBN{979-8-4007-2816-7/2026/11}

\begin{document}

%\title{Carbon-Aware Scaling for Elastic Batch Workloads using Capacity Variations}
%\title{CCS: Carbon-Aware Capacity Scaling for Elastic Batch Workloads}
%\systemName: Carbon-Aware Scheduling for Elastic Clusters and Workloads with Unknown Job Length}
%\title{CarbonFlex: Enabling Carbon-aware Provisioning and Scheduling for Cloud Clusters}
\title[Flex: Grid-Responsive Provisioning and Scheduling for Elastic Cloud Clusters]{Flex: Grid-Responsive Provisioning and Scheduling \\ for Elastic Cloud Clusters}

\begin{abstract}
    \input{sections/0-abstract}

\end{abstract}



\maketitle

\section{Introduction} \label{sec:intro}
\input{sections/1-intro}

\section{Background} \label{sec:background}
\input{sections/2-background}

% \note{BG is over size}
% \clearpage

\section{Grid-Responsive Elastic Cloud Clusters} \label{sec:prob}

\input{sections/2-problem-new}

% \section{Offline Optimal} \label{sec:offline}
% \input{sections/3-offline}

\section{\systemName Design} \label{sec:des}

\input{sections/3-design-new}

% \note{Design is over size}
% \clearpage
%\vspace{-0.1cm}
\section{\systemName Implementation} \label{sec:imp}
\input{sections/4-implementation}

% \note{Implementation is over size}
% \clearpage
\section{Experimental Evaluation} \label{sec:eval}

\input{sections/5-eval}

% \section{Discussion} \label{sec:discussion}
% \input{sections/6-discussion}

\section{Related Work} \label{sec:related}

\input{sections/6-related}

\section{Conclusion} \label{sec:conclusion}
\input{sections/7-conclusion}

\begin{acks}
We thank the anonymous SoCC reviewers for their valuable insight and feedback. This research is supported by National Science Foundation (NSF) grants 2213636, 2211302, 2211888, 2325956, 2105494 and U.S. Department of Energy award DE-EE0010143.
This work used CloudBank, which is supported by NSF grant 19250001.
\end{acks}

%\clearpage
\bibliographystyle{ACM-Reference-Format}
\bibliography{main}

%\clearpage
%\onecolumn
\appendix
\appendixpage

\input{sections/9-theory-appendix}

\end{document}

%% file: helpers.tex
\usepackage{enumitem}
\usepackage{pifont}% http://ctan.org/pkg/pifont
\usepackage{makecell}
\usepackage{hyperref}  % Load this last!
 \usepackage[export]{adjustbox}
\usepackage{appendix}
\newtheorem{remark}{Remark}

\newtheorem{lemma}{Lemma}

\def\ie{\textit{i.e.},\xspace}
\def\eg{\textit{e.g.},\xspace}

\newcommand{\J}{\mathcal{J}}
\newcommand{\T}{\mathcal{T}}
\newcommand{\C}{\mathcal{C}}

\newcommand{\cmark}{\textcolor{black}{\ding{52}}}%
\newcommand{\xmark}{\textcolor{red}{\ding{56}}}%

\def\intensityunit{g$\cdot$CO$_2$eq/kWh\xspace}
\def\priceunit{\textcent/kWh\xspace}

\def\systemName{\textit{Flex}\xspace}
\def\systemNameSim{\textit{Flex-Simulator}\xspace}
\def \systemOracle{\textit{Flex (Oracle)}\xspace}

\def \agnostic{\textit{Agnostic}\xspace}
\def \GAIA{\textit{GAIA}\xspace}
\def \WaitAWhile{\textit{WaitAWhile}\xspace}
\def \CarbonScaler{\textit{CarbonScaler}\xspace}
\def \GREEN{\textit{GREEN+}\xspace}
\def \CarbonScalerOracle{\textit{CarbonScaler\xspace(Oracle)}\xspace}
\def \CarbonScalerMax{\textit{CarbonScaler(Oracle, NoLimit)}\xspace}

\newcounter{todonumber}
\stepcounter{todonumber}

\newtheorem{theorem}{Theorem}

%% file: sections/0-abstract.tex
Computing energy demand and grid capacity growth are at a crossroads.
Data centers are becoming large, concentrated electricity consumers just as power grids face capacity shortages while also integrating more variable renewable generation, causing grid conditions to change significantly across hours and seasons.
These variations create an opportunity for \emph{grid-responsive} operations, in which workloads are shifted to align energy use with time-varying grid conditions, prioritizing clean, low-cost periods of surplus generation and low congestion.
Fortunately, a significant fraction of computing demand comes from batch jobs that are delay-tolerant and elastic, enabling clusters to shift or scale jobs, or adjust the provisioned capacity in response to grid conditions.
However, provisioning and scheduling decisions are closely coupled. Effective cluster capacity depends on how jobs are scaled, while job-scaling decisions depend on the available capacity. At the same time, exogenous grid signals (\eg carbon intensity and electricity prices) determine the cost of both scheduling and provisioning decisions.

To address this problem, we present \systemName{}, a grid-responsive resource manager for elastic cloud clusters that coordinates provisioning of cluster capacity and scheduling of elastic batch jobs.
\systemName{} takes a data-driven approach, periodically learning optimal scheduling and provisioning decisions based on recent history and reusing them to guide coherent runtime provisioning and scheduling.
We implement \systemName{} on top of AWS ParallelCluster and evaluate it using real-world workload traces and grid signals, including carbon intensity and electricity prices.
Our results show that \systemName reduces carbon emissions by up to $\sim$57\%, within 2.1\% of the oracle, and electricity costs by up to 60\% relative to signal-agnostic baselines.

%% file: sections/1-intro.tex
Data centers' energy demand is growing at unprecedented rates~\cite{Shehabi2024:USDCReport}, with estimates that the global data center energy consumption will reach 1000 terawatt-hours (TWh) by 2026~\cite{iea2024electricity}, and 6--12\% of the total U.S. electricity demand within the next 3--5 years~\cite{Shehabi2024:USDCReport}. This surging demand is causing severe energy shortages in the grid~\cite{Lin2024:ExplodingAI}, and raising concerns about energy prices and the environmental and health impacts of data centers~\cite{han2024unpaidtollquantifyingpublic, Maji2025:DCEmissions}. 
To address this issue, cloud providers have relentlessly focused on improving the energy efficiency of their data centers through infrastructure- and operational-level optimizations~\cite{energy-efficiency-survey, Barroso:2007:energy-proportional}. 
Unfortunately, energy efficiency improvements alone are insufficient to offset the significant growth in energy consumption~\cite{Koomey2015:moore}. Efficiency improvements reduce how much energy a data center uses, but do not change
when it is drawn from an increasingly variable grid.
This variability has motivated the design of \emph{grid-responsive} data centers, which align energy consumption with time-varying grid conditions, prioritizing clean, low-cost periods of surplus generation and low congestion~\cite{Lin2023:Adapting, Colangelo2026:grid-interactive, Wierman2014:DC-DemandResponse}.

Grid-responsive data center operations react to changes in grid conditions to enhance a data center's operational efficiency while also helping the grid~\cite{Qureshi2009:Cutting, Sukprasert2024:Limitations}.
To do so, data centers typically respond to \emph{exogenous grid signals} to adjust their energy consumption. For instance, data centers respond to \emph{binary} demand response (DR) signals~\cite{Zhang2019:DRParticipation, Acun2026:Flexibility}, which force data centers to curtail their on-grid energy consumption, as well as \emph{continuous} cost signals, such as carbon intensity, electricity prices, and renewable availability, that determine the cost or emissions associated with that energy consumption.
Given the potential above, there has been significant recent work on \emph{signal-aware} optimizations to reduce costs and emissions of cloud systems by workload scheduling, scaling, and resource provisioning decisions to align the energy consumption with the availability of low-cost or green energy sources~\cite{Hanafy2023:CarbonScaler, Wiesner2021:WaitAwhile, Souza2023:Ecovisor, Acun2023:CarbonExplorer, Dodge2022:AICloud, Gsteiger2024:Caribou, Rao2010:Multi-Market, Goiri2012:GreenHadoop, Qureshi2009:Cutting, Liu2011:GeoBalancing}. 
These efforts focus on batch jobs, such as machine learning training~\cite{pytorch} and HPC workloads~\cite{Fox2017:E-HPC, Martin2024:Proteo, Tarraf2024:Malleability}, which are often delay-tolerant and elastic, enabling data centers to shift and scale this work toward the hours when electricity is cleaner and cheaper, trading modest delays for lower costs and emissions~\cite{Acun2023:CarbonExplorer, Sukprasert2024:Limitations, Radovanovic2023:VCCPaper, Hanafy2023:CarbonScaler, Xu2025:Green, Lechowicz2026:SPADE}.

Existing grid-responsive resource management approaches focus either on \emph{workload scheduling} by deciding when and how to execute batch jobs, whether for individual or multiple jobs, or on \emph{capacity provisioning} by adjusting the number of cloud-provisioned CPU or GPU servers according to the grid signal.
Despite these two decisions being tightly coupled, existing systems do not coordinate their decision-making for these two tasks, treating them as separate and independent problems.
For instance, multiple approaches focus on \emph{scheduling} a single job by adapting the demand against grid signals~\cite{Wiesner2021:WaitAwhile, Hanafy2023:CarbonScaler, Lechowicz2023:DTPR}. However, such approaches often require \emph{a priori} knowledge of job length and treat the cloud as an unlimited pool from which resources can always be drawn, which can lead to the ``thundering herd'' problem~\cite{ruane1990:THPrblem, Hanafy2026:Untangling}. Cluster-level schedulers apply the same ideas across many jobs by considering a fixed cluster capacity and shifting and scaling jobs to align demand with the hours of cheap and low-carbon energy~\cite{Hanafy2024:GoingGreen, Xu2025:Green, Goiri2012:GreenHadoop}.
In contrast to scheduling approaches, a second line of work focuses on resource \emph{provisioning}, exploiting the low average utilization of data centers~\cite{Tirmazi2020:Borg, Shehabi2016:USDCReport, Shehabi2024:USDCReport} to scale the resource pool up and down according to the grid signal~\cite{Lin2023:Adapting, Zhang2021:VariableCapacityChallanges, Zheng2020:Curtailment}. For example, Google's virtual capacity curve (VCC) forecasts future demand to compute a time-varying capacity limit, implicitly pushing flexible work toward low-carbon periods and scheduling jobs within the varied capacity using signal-agnostic policies~\cite{Radovanovic2023:VCCPaper}.
However, provisioning and scheduling decisions are highly co-dependent and tightly coupled with the exogenous grid conditions.

Designing grid-responsive capacity provisioning and scheduling approaches must address three key challenges:
\emph{\textbf{Challenge 1:} Coupling between capacity and job scalability.}
Expanding capacity is beneficial only if the available jobs can efficiently convert that capacity into useful progress. At the same time, how jobs scale depends on the capacity available to them, while current and future grid signals determine the cost of both provisioning and scheduling decisions.
\emph{\textbf{Challenge 2:} Intertemporal coupling across the scheduling horizon.}
This coupling extends across the entire planning horizon. Every deferral or scale-down shifts demand into later time slots. So, the future demand that a cluster must provision depends not only on exogenous workload arrivals and grid conditions, but also on the demand induced by its own earlier scheduling decisions. As a result, grid-responsive operation is inherently sequential, where each decision changes the workload state and reshapes the feasible and cost-effective options available later.
\emph{\textbf{Challenge 3:} Multiple decision timescales.}
Scheduling and provisioning decisions operate on different timescales. Capacity adjustments typically require minutes to hours, whereas scheduling decisions may be made at the second granularity. A grid-responsive controller must therefore coordinate fast workload-level decisions with slower infrastructure-level actions without overreacting to short-term fluctuations or missing longer-term grid opportunities.

This coupling across decisions and time makes wrong decisions costly on both fronts. Shrinking capacity at the wrong moment strangles running jobs or forces delays, while growing it prematurely leaves servers idle or assigns them to jobs that scale poorly, hurting completion times and increasing costs.
More importantly, layering a signal-aware capacity curve over a signal-agnostic scheduler will lead to suboptimal and even counterproductive decisions, since whether acquiring or releasing a server is worthwhile depends jointly on the grid signal and on the elasticity and slack of the jobs that would use it, a result we confirm empirically in \autoref{sec:eval} (See \autoref{fig:vcc_compare}).
As a result, the resource provisioning and scheduling decisions must simultaneously consider (i) the cluster's current and expected demand and its elasticity, (ii) the exogenous grid signal and its variations, and (iii) the trade-off between delaying work and the savings that waiting may yield.

To address these challenges, we present \systemName, a grid-responsive resource manager that jointly provisions cluster capacity and schedules elastic batch jobs. \systemName builds on a key observation: while these dependencies make it hard to decide any dimension in isolation, there is one place where all of them are fully visible at once, \emph{the realized past}. Over the past time windows, the grid signal, job arrivals, their scaling behavior, and all interactions among them are known exactly, so the jointly optimal provisioning and scheduling decisions can be computed with certainty. And because these inputs recur with diurnal and weekly patterns, decisions that were optimal in the recent past, or in similar situations, remain effective in the near future. Our hypothesis is that \emph{recomputing the optimal provisioning and scheduling decisions over the recent past, and reusing them at runtime, is an effective approach for grid-responsive management of elastic batch workloads}.

\systemName operationalizes this intuition with a \emph{data-driven} approach inspired by {\em hindsight learning}~\cite{Sinclair2023:Hindsight}, replaying the recent past via an offline algorithm that knows exactly when jobs arrived, their elasticity and delay tolerance, and how the signal evolved. 
However, computing this optimal schedule is nontrivial, and we address this complexity by deriving a fast, polynomial-time algorithm that computes the optimal joint provisioning and scheduling decisions.
\systemName then distills these executions into simple provisioning and scheduling runtime parameters, a resource provisioning capacity, and a threshold for marginal progress, enabling jobs to scale up until their added progress no longer justifies the signal's cost. Then, at runtime, \systemName matches the cluster's current state against the recomputed executions and reuses their parameters, keeping provisioning and scheduling coherent without explicit coordination, and falling back to conservative decisions whenever the present no longer resembles the past.

Importantly, by following the offline optimal decisions, \systemName encodes the dependencies among provisioning and scheduling decisions, the exogenous grid signal and its variations, and the sequential nature of the decisions. As a result, \systemName addresses the key limitations of prior work and operates without (i) \emph{a priori} knowledge of job lengths, since the threshold prioritizes jobs by their observed marginal progress rather than their declared duration, and (ii) forecast-based planning over future demand and grid signals, since it reuses decisions from the realized past, which embeds how demand responds to provisioning and scheduling decisions.

In designing, implementing, and evaluating \systemName, our paper makes the following contributions.

\begin{enumerate}[leftmargin=*]
    \item We present the design of \systemName, a grid-responsive resource manager that coordinates cluster capacity provisioning and scheduling of elastic batch jobs using hindsight learning---continuously relearning optimal decisions based on recent history and reusing them at runtime.
    
    \item We implement \systemName on top of AWS ParallelCluster~\cite{aws-pc} and evaluate it on CPU and GPU clusters running elastic MPI and machine learning workloads, complemented by year-long trace-driven simulations.
    \item Our evaluation spans real-world workload traces and grid signals across eight grid regions, and compares \systemName against signal-agnostic, per-job scaling, fixed-capacity, and provisioning-only baselines, as well as an oracle with perfect knowledge of future signals and job lengths. 
    Using carbon intensity as the grid signal, \systemName{} reduces energy-related emissions by up to $\sim$57\% compared to signal-agnostic baselines, coming within 2.1\% of an oracle with perfect future knowledge, and cuts energy costs by up to 60\% under the electricity-price grid signal.
\end{enumerate}

%% file: sections/2-background.tex
This section presents background on data centers and the electricity grid, grid-responsive resource management, and elastic batch jobs.

\begin{figure}[t]
    \centering
    \includegraphics[width=0.9\linewidth]{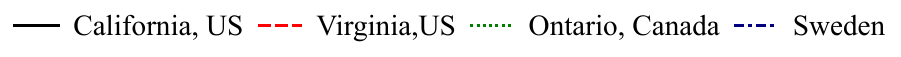}\\
    \hfill
    \begin{subfigure}[b]{0.48\linewidth}
   \includegraphics[width=\linewidth]{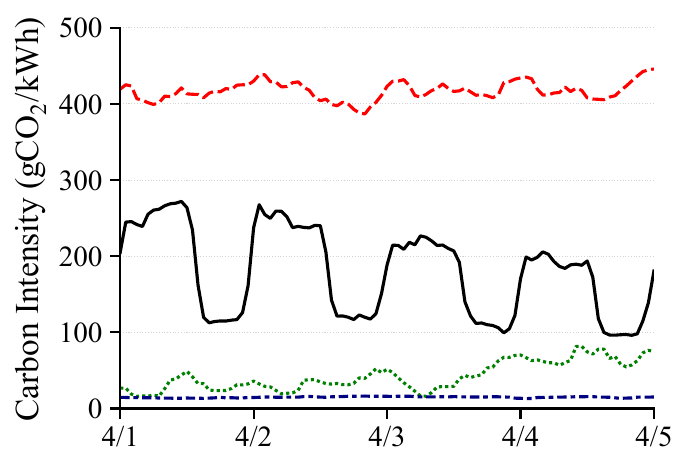}
   \caption{Carbon Intensity}
    \label{fig:trace_week_CI}
    \end{subfigure}
    \hfill
     \begin{subfigure}[b]{0.48\linewidth}
   \includegraphics[width=\linewidth]{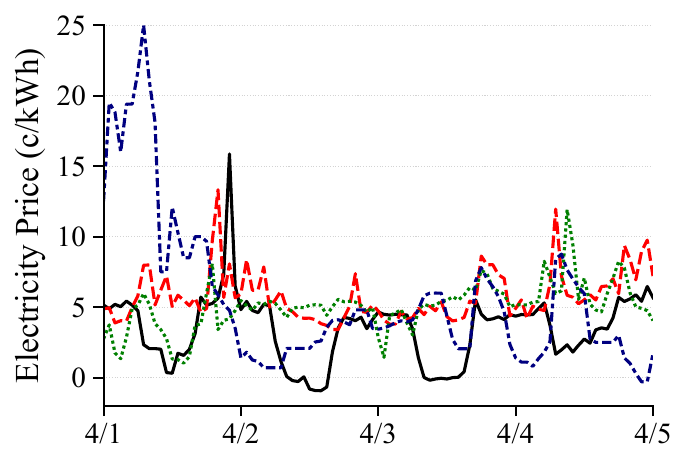}
   \caption{Cost}
    \label{fig:trace_week_cost}
    \end{subfigure}
    \hfill
    \hfill
      \caption{Energy's carbon intensity and cost variations in four locations in the first week of April 2022.}
    \Description{Energy's carbon intensity and cost variations in four locations in the first week of April 2022.}
\label{fig:trace_week}
\end{figure}

\subsection{Grid-Responsive Data Centers}\label{sec:background_demand_response}
Data centers are major contributors to load in most electricity grids~\cite{Colangelo2026:grid-interactive}. 
To maintain safe and reliable grid operation, utilities and system operators have traditionally relied on \emph{binary} demand-response (DR) signals, where large consumers are asked or incentivized to adjust their electricity consumption in response to grid needs~\cite{Acun2026:Flexibility, peak-paper, Liu2013:Peak, Wierman2014:DC-DemandResponse, Colangelo2026:grid-interactive}.
Depending on the objective, DR may require data centers to curtail power during peak-demand or emergency periods, or encourage them to reduce consumption through economic signals such as Coincident Peak Pricing (CPP)~\cite{Liu2013:Peak, Wierman2014:DC-DemandResponse, Acun2026:Flexibility}.

In addition to DR signals, grid operators have developed a range of \emph{continuous} grid signals that capture variations in energy supply and demand, such as carbon intensity, electricity prices, and renewable availability. 
These signals are often used as proxies for the energy supply mix and for the availability of surplus generation and low congestion.
For instance, energy's carbon intensity (CI), measured in grams of CO$_2$ equivalent per kWh of electricity (\intensityunit), captures the greenhouse gas (GHG) emissions per unit of electricity generated.
The emissions from a unit of electricity generated depend on the mixture of generation sources, where fossil-based sources (\eg natural gas and coal) have a high carbon intensity, while renewable sources (\eg wind, solar, and hydro) have a low or zero carbon intensity.

The presence of such signals has motivated a broader design goal of \emph{grid-responsive} data centers. Rather than reacting only to discrete DR events, a grid-responsive data center continuously aligns its workload execution, capacity provisioning, and scheduling decisions with time-varying grid conditions. By shifting flexible demand toward periods with lower carbon intensity, lower electricity prices, higher renewable availability, or lower grid congestion, data centers can reduce emissions and operating costs while also providing more predictable and flexible demand to the grid~\cite{Lin2023:Adapting, Colangelo2026:grid-interactive, Wierman2014:DC-DemandResponse, zero-carbon}.

\subsection{Carbon Intensity and Cost Signals}\label{sec:background_CI}
Diurnal variations in energy supply and demand cause grid conditions, such as carbon intensity and electricity prices, to vary significantly across hours and seasons. \autoref{fig:trace_week} shows four regions with different energy sources, illustrating variations in carbon intensity and energy costs. \autoref{fig:trace_week_CI} highlights that the carbon intensity varies widely among locations, with up to a $\sim$400\intensityunit difference between Virginia and Canada, Ontario, two regions equally distant from customers in the northeast of the US.  Moreover, the figure shows that even at a given location, carbon intensity fluctuates over time. For instance, the figure shows that the carbon intensity in California varies daily by $\sim$100\intensityunit, resulting in different carbon footprints as per the time of execution. 

Similarly, \autoref{fig:trace_week_cost} shows comparable variations in cost, with up to 24\priceunit across grids (Sweden vs. California), and up to 18\priceunit between the cheapest and most expensive hours in a single day in Sweden. Lastly, it is worth noting that while all regions exhibit temporal variation in both carbon and cost, the two signals correlate differently across regions. For instance, in California, both signals track the availability of solar energy, whereas in Sweden, carbon intensity remains relatively stable, as it depends heavily on hydroelectric sources, yet its cost still varies substantially with changes in demand.

\input{tables/table_summary2}

\subsection{Signal-aware Scheduling and Provisioning}
Grid-responsive scheduling has focused on approaches that schedule jobs against the grid signal (\eg carbon intensity and electricity prices), and on provisioning that adjusts the cluster size accordingly.

\noindent\emph{\textbf{Temporal Shifting.}}
The temporal variations in carbon intensity have motivated researchers to utilize the inherent temporal flexibility of batch jobs by running them during low-carbon periods and suspending them during high-carbon periods~\cite{Wiesner2021:WaitAwhile, Souza2023:Ecovisor, Hanafy2024:GoingGreen, Acun2023:CarbonExplorer, Lechowicz2023:DTPR, Sukprasert2024:Limitations}.
The same idea extends to electricity price, where schedulers shift batch jobs toward cheaper hours~\cite{Goiri2011:GreenSlot}.
In addition, researchers have proposed elastic scheduling methods in which jobs are typically scaled down during low-carbon periods and suspended during high-carbon periods, eliminating the need to extend the expected completion time and increasing savings compared to typical suspend-resume approaches~\cite{Hanafy2023:CarbonScaler, Souza2023:Ecovisor}. 
The key issue is that these approaches schedule individual jobs, relying on hand-tuned thresholds~\cite{Wiesner2021:WaitAwhile, Souza2023:Ecovisor} that require manual effort to balance savings against performance, or on full knowledge of job lengths~\cite {Hanafy2023:CarbonScaler, Lechowicz2023:DTPR}, which is error-prone in practice~\cite{Kuchnik2019ThisIW, Ambati2021:GoodThings}.
They also ignore cluster-wide capacity constraints, causing demand bursts during low-carbon periods~\cite{Hanafy2026:Untangling}, a phenomenon known as the stampede or thundering-herd problem.

\noindent\emph{\textbf{Capacity Provisioning.}} At the cluster level, grid-responsive optimizations utilize the ability to vary the cluster capacity as well as the low average utilization of data centers~\cite{Tirmazi2020:Borg, Shehabi2016:USDCReport, Shehabi2024:USDCReport}, changing the cluster capacity based on temporal variations in the grid signal, predominantly carbon intensity~\cite{Radovanovic2023:VCCPaper, Lin2023:Adapting, Perotin2023:Risk, Zheng2020:Curtailment, Zhang2021:VariableCapacityChallanges}.
The key idea behind these approaches is to vary cluster size based on electricity carbon intensity, opportunistically using larger cluster capacities during low-carbon periods.
For example, Google's virtual capacity curve (VCC)~\cite{Radovanovic2023:VCCPaper} computes a time-varying capacity limit and schedules jobs within it, while~\cite{Lin2023:Adapting} scales capacity based on electricity prices.
However, these approaches vary capacity without leveraging application elasticity or workload characteristics, leading to higher costs, emissions, or delays.

~\autoref{tab:bg_summary} summarizes these grid-responsive scheduling and provisioning approaches. Across both lines, scheduling approaches operate within a fixed capacity while provisioning approaches ignore how jobs scale, so none co-designs the two. 
In this paper, we explore grid-responsive provisioning and scheduling that jointly vary cluster capacity and schedule elastic jobs, reducing carbon emissions or electricity cost.

\begin{figure}[t]
    \centering
        \includegraphics[width=\linewidth]{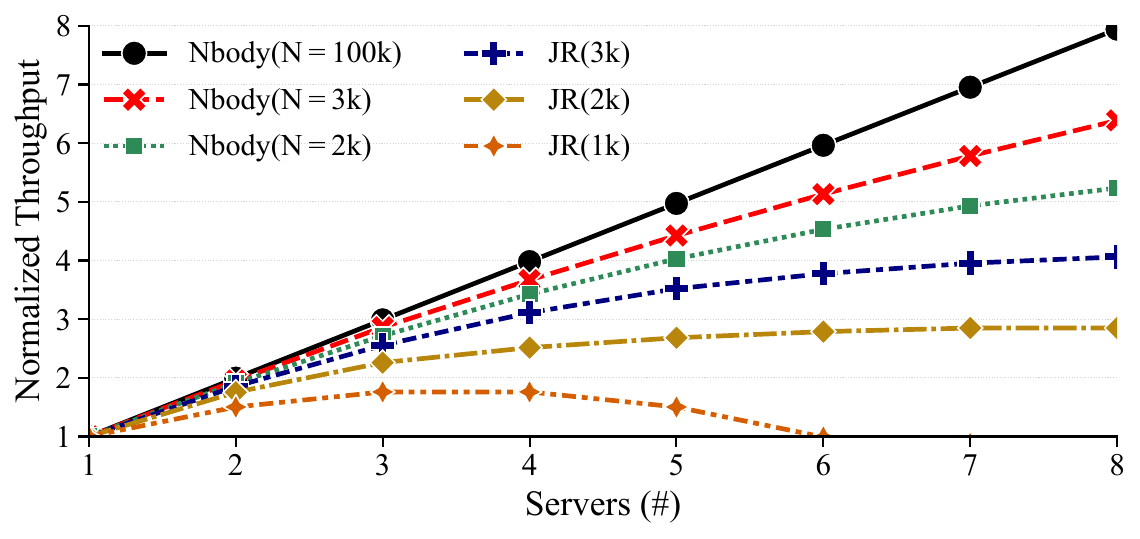}
   \caption{Scaling profiles of different MPI and ML jobs.
   }%
   \Description{Scaling profiles of different MPI and ML jobs.}%
    \label{fig:elasticity}
\end{figure}

\subsection{Elastic Batch Jobs and Scaling Profiles}
Elastic scaling, also referred to as malleability~\cite{Tarraf2024:Malleability, Martin2024:Proteo}, is the ability to change the allocated resources seamlessly and has been shown to be applicable to a broad class of distributed batch jobs. 
For instance, machine learning frameworks (\eg PyTorch~\cite{pytorch}), Data Processing frameworks (\eg Spark~\cite{spark}), Parallel Programming Frameworks (\eg MPI~\cite{Isa2016:ElasticMPI} and Charm++\cite{kale1993charmpp}) allow applications to adapt resources dynamically.      
Elastic scaling capabilities have enabled cluster operators to
avoid head-of-line blocking, improve fault tolerance, and reduce energy consumption~\cite{Tarraf2024:Malleability, Gupta2014:RealizingMalleable, Prabhakaran2015:Malleable, Xiao2020:AntMan, Amico2019:Slowdown, Peng2018:Optimus, Subramanya2023:Sia, Qiao2021:Pollux}. In contrast, we exploit elastic scaling to make clusters grid-responsive, reducing their carbon emissions or electricity cost. 

Elastic scaling of a batch job must consider its scaling characteristics, since batch jobs rarely scale linearly with the number of allocated servers. Typically, the scaling behavior of a distributed job depends on its compute and communication characteristics~\cite{Subramanya2023:Sia, Hanafy2023:CarbonScaler, Li2024:CommCCS}. The greater the communication per unit of compute, the less the job's throughput scales, as communication bottlenecks grow with added resources, resulting in diminishing increases in performance.
\autoref{fig:elasticity} shows the elastic scaling profiles of different batch applications.  The setup for these profiles is detailed in ~\autoref{sec:eval_setup}.
As shown, applications exhibit varying elastic scaling behaviors as per their compute-to-communication ratios. For example, EffNet-S has 8.37 GFLOPs and is 82.7 MB, while ResNet18 has 1.81 GFLOPs and is 44.7 MB~\cite{pytorch}, making the communication (model memory footprint) per unit compute 9.8MB/GFLOPs and 24.6MB/GFLOPs for EffNet-S and ResNet18, respectively, yielding higher scalability of EffNet-S as depicted in the figure.

%% file: tables/table_summary2.tex
\begin{table}[t]
\caption{Comparison of prior approaches and \systemName.}
\label{tab:bg_summary}
\resizebox{\linewidth}{!}{%
\begin{tabular}{l|c|c|c|c|c}
\toprule
\textbf{Approach} & \makecell{\textbf{Multiple} \\ \textbf{Jobs}} & \makecell{\textbf{Unknown} \\ \textbf{Job Length}}  &\makecell{\textbf{Capacity} \\ \textbf{Scaling}} & \makecell{\textbf{Signal-aware} \\ \textbf{Scheduling}} & \makecell{\textbf{Resource} \\ \textbf{Scaling}} \\ \midrule
WaitAWhile~\cite{Wiesner2021:WaitAwhile} & \xmark & \cmark  &\xmark & \cmark & \xmark \\
DTPR~\cite{Lechowicz2023:DTPR} & \xmark & \xmark  &\xmark & \cmark & \xmark \\
Wait and Scale~\cite{Souza2023:Ecovisor} & \xmark & \cmark  &\xmark & \cmark & \cmark \\
Carbon Scaler~\cite{Hanafy2023:CarbonScaler} & \xmark & \xmark  &\xmark & \cmark & \cmark \\ \bottomrule \toprule
GAIA~\cite{Hanafy2024:GoingGreen} & \cmark & \cmark  &\xmark & \cmark & \xmark \\
Green~\cite{Xu2025:Green} & \cmark & \cmark  &\xmark & \cmark & \cmark \\
Risk-Aware ~\cite{Perotin2023:Risk} & \cmark & \xmark  &\cmark & \cmark & \xmark \\
Google VCC~\cite{Radovanovic2023:VCCPaper} & \cmark & \cmark  &\cmark & \xmark & \xmark \\
Adaptive Capacity~\cite{Lin2023:Adapting} & \cmark & \cmark  &\cmark & \xmark & \xmark \\ \bottomrule \toprule
\systemName & \cmark & \cmark  &\cmark & \cmark & \cmark \\ \bottomrule
\end{tabular}%
}
\end{table}

%% file: sections/2-problem-new.tex
This section describes the grid-responsive provisioning and scheduling problem in elastic cloud clusters that this paper addresses. 
Our work targets a slotted horizon $t\in \T \triangleq \{1,\dots,T\}$ and a set of preemptive jobs $\J$, with $ N=|\J|$, that run concurrently on multiple servers and communicate across components during execution. 
The cloud cluster is homogeneous, consisting of either CPU or GPU servers, and its capacity can be varied over time (using cloud interfaces to acquire or release servers) up to a configurable maximum $M$.
Users submit jobs to one of several queues (\eg by priority), each with a pre-configured maximum delay, or ``slack,'' chosen according to their willingness to delay in exchange for higher savings.
Each job $j$ has an arrival time $a_j\in\T$, a length $l_j$, and a best-effort delay target $d_j$ of its selected queue $i$. 
In the offline setting, however, job lengths are known, so we map this delay budget to a completion deadline $e_j = a_j+l_j+d_j$ that the ILP enforces.
At each time slot, each job $j$ can use a number of non-shared resources $k$ allocated to the job that can be varied at run-time between an upper and lower bound:  
$k\in [k_j^{\min}, k_j^{\max}]$ where $k_j^{\min}$ and  $k_j^{\max}$ denote the minimum and maximum numbers of servers that can be allocated to that job. 
Our model also supports non-elastic workloads (i.e., $k_j^{\min}=k_j^{\max}$), in which only the cluster capacity is scaled (see ~\autoref{sec:eval_ccs}).

Each job $j$ has a \emph{normalized} elastic scaling throughput profile $\Phi_j:k\in\mathbb{Z}_{\ge0}\to\mathbb{R}_{\ge0}$ with $\Phi_j(0)=0$ and non-increasing marginals over $k\ge1$, $\Delta_j(k)\triangleq \Phi_j(k)-\Phi_j(k-1)$, normalized so that $\Delta_j(k_j^{\min}) = 1 \ge \Delta_j(k_j^{\min}+1)\ge \cdots\ge 0$.
This assumption is realistic for elastic jobs, as adding more resources will not introduce superlinear behavior but rather diminishing returns~\cite{Hanafy2023:CarbonScaler, Subramanya2023:Sia}.
In practice, the profiles can be learned from profiling or performance~\cite{paleo, Oyama2016, Pei2019, justus2018, cai2017neuralpower, Peng2018:Optimus, performance_modeling}.
Using a resource at time slot $t$ incurs a cost $c_t$ set by an exogenous \emph{grid signal}, where $c = \{c_1, c_2, \dots, c_T\}$ is the sequence of signal values over the $T$-length window.
We formalize the offline problem as an integer linear program (ILP).

\noindent\textbf{Decision variables.}
$x_{j,t,k}\in\{0,1\}$: whether the $k$-th resource unit for job $j$ is active at time $t$, 
$y_{j,t}\in\{0,1\}$: whether job $j$ is running in slot $t$.\footnote{We model allocation at the granularity of individual units ($x_{j,t,k}$) rather than an integer count, so the concave throughput enters linearly and each unit carries an explicit cost-per-benefit ratio $c_t/\Delta_j(k)$, which we use to drive our offline optimality (\autoref{theorem:CFO}) and the runtime marginal-progress threshold (\autoref{sec:des}).}

\noindent\textbf{Objective.}
{
\begin{align}
  \min  \C(x)= \quad \sum_{t\in\T}\sum_{j\in\J} c_t \sum_{k=1}^{k_j^{\max}} x_{j,t,k}.
  \label{eq:objective}
\end{align}
}

\noindent\textbf{Constraints.}
{
\begin{align}
  &\sum_{t=a_j}^{e_j}\sum_{k=1}^{k_j^{\max}} \Delta_j(k)\,x_{j,t,k} \ge l_j
  && \forall j\in\J, \label{eq:work}\\
  &\sum_{j\in\J}\sum_{k=1}^{k_j^{\max}} x_{j,t,k} \le M
  && \forall t\in\T, \label{eq:cap}\\
  &x_{j,t,k} \ge x_{j,t,k+1}
  && \forall j,\, t\in[a_j,e_j],\, k<k_j^{\max}, \label{eq:stair}\\
  &x_{j,t,k} \le y_{j,t}
  && \forall j,\, t\in[a_j,e_j],\, k\le k_j^{\max}, \label{eq:xleq}\\
  &k_j^{\min} y_{j,t} \le \sum_{k=1}^{k_j^{\max}} x_{j,t,k} \le k_j^{\max} y_{j,t}
  && \forall j,\, t\in[a_j,e_j], \label{eq:slotbounds}\\
  &x_{j,t,k}=0,\ y_{j,t}=0
  && \forall j,\, t\notin[a_j,e_j],\, k. \label{eq:windowfix}
\end{align}
}
We assume there is a feasible solution for the given instance.

\input{tables/features}

\noindent\textbf{Interpretation.}
\eqref{eq:objective} represents the objective function, where in a slot $t$, allocating $s = \sum_{k=1}^{k_j^{\max}} x_{j,t,k}$ resources to job $j$ (with $k_j^{\min}\le s\le k_j^{\max}$) costs $c_t s$. At each time slot, using $s$ resources produces $\Phi_j(s)=\sum_{k=1}^s \Delta_j(k)$ units of work for job $j$, where \eqref{eq:work} ensures that the job is assigned enough resources within $[a_j, e_j]$ to complete its required work (equivalently, its length $l_j$).
\eqref{eq:cap} ensures that cluster capacity is respected. \eqref{eq:stair} ensures resources are assigned incrementally. \eqref{eq:xleq} ensures resources are assigned only if the job is running at $t$. \eqref{eq:slotbounds} ensures that resources are assigned within the upper and lower limits. \eqref{eq:windowfix} ensures that jobs run within their allowed time.

At runtime, however, the future grid signal, arrivals, and job lengths are unknown. Each slot, \systemName must still decide how much capacity to hold (\emph{provisioning}) and how many servers to give each queued or running job (\emph{scheduling}), directing capacity to the jobs that turn it into the most progress while respecting each job's slack. As highlighted in \autoref{sec:intro}, the two are tightly intertwined, since capacity pays off only when queued jobs can scale into it, each deferral or scale-down reshapes the future queue, \emph{indirectly coupling the decisions across time}, and the two act on different timescales. Rather than solve this coupling online or plan against a forecast of future demand and grid signals, \systemName takes a \emph{data-driven} approach that recomputes the offline-optimal decisions over the realized past and reuses them at runtime, distilled into separate provisioning and scheduling policies, each executing at its own timescale.

%% file: tables/features.tex
% \begin{figure*}[ht]
%     \centering
%     \hfill
%     \begin{minipage}[t]{0.6\linewidth}
\begin{figure*}[t]
    \centering
        \includegraphics[width=0.8\linewidth]{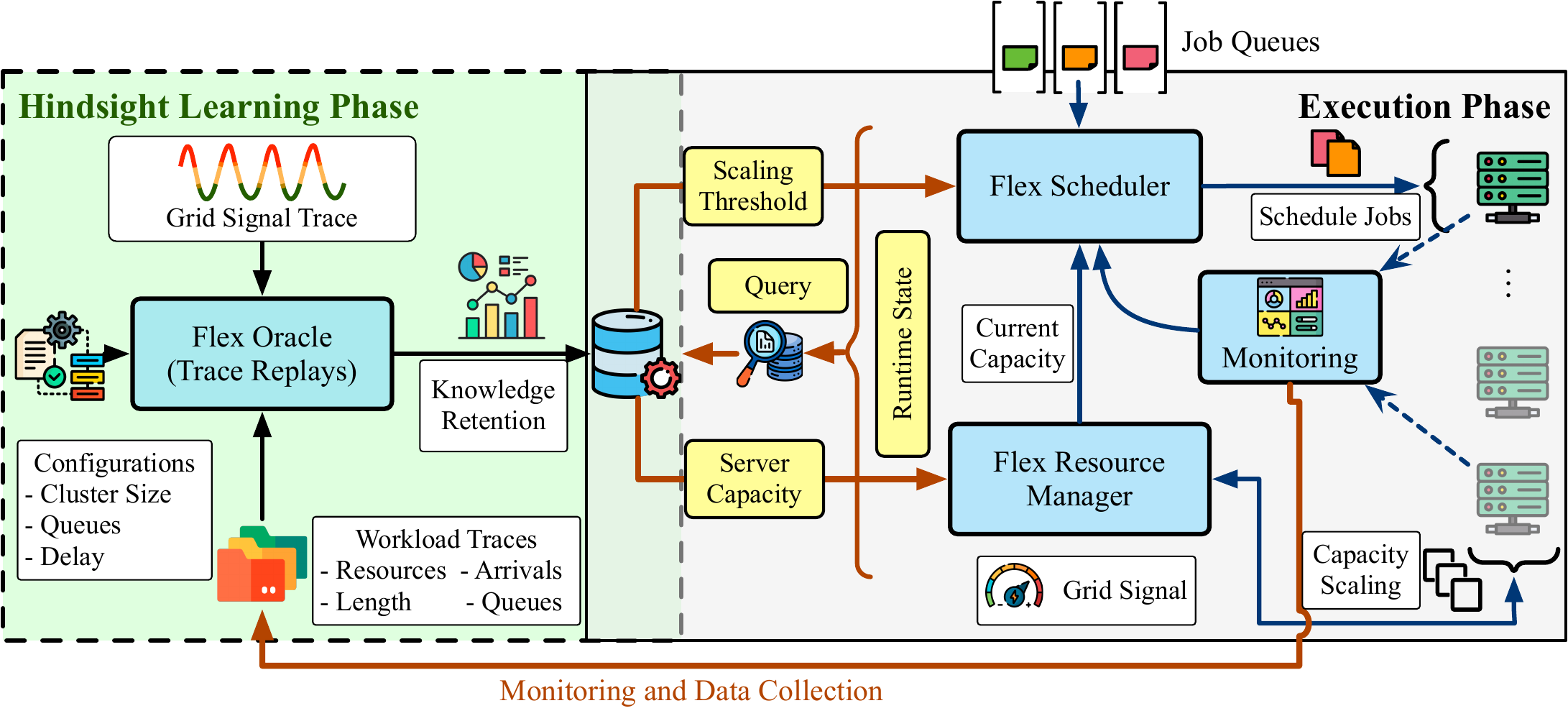}
    \caption{Overview of learning and execution phases of \systemName. }
    \Description{Overview of \systemName.}
    \label{fig:system_overview}
\end{figure*}
%     \end{minipage}
%     \hfill

%     \hfill
%     \hfill

% \end{figure*}

\begin{figure}[t]
    \centering
    \includegraphics[width=\linewidth]{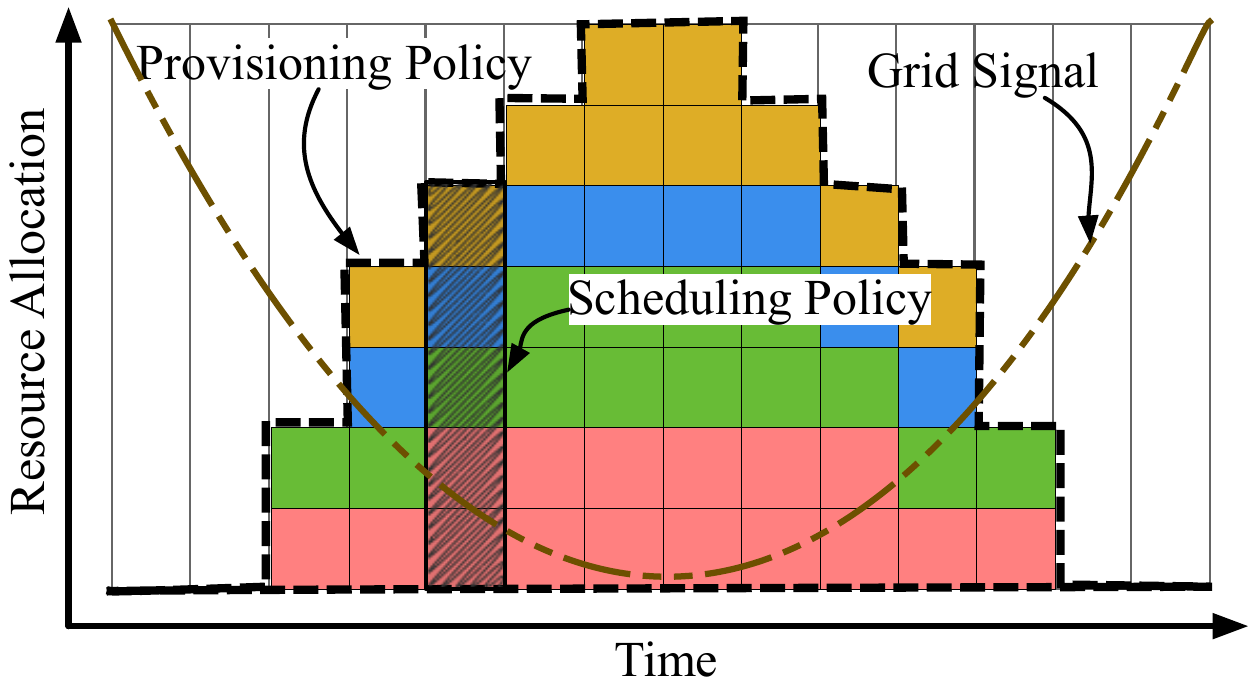}
    \caption{Representing the decisions of \systemOracle %in relation to the load and carbon intensity and how to represent them 
    as a provisioning and scheduling policy.}
    \Description{Representing the decisions made by \systemName Oracle.}
    \label{fig:idea}
\end{figure}

%% file: sections/3-design-new.tex
This section presents \systemName's design and its elastic scaling-driven provisioning and scheduling algorithms.

\subsection{\systemName Overview}
\systemName is a grid-responsive resource manager for batch-oriented cloud clusters. The design of \systemName is based on three key principles:
\begin{enumerate}[leftmargin=*]
    \item \textbf{Elastic Scaling:} \systemName applies elastic scaling to both its provisioning and scheduling policies, with a grid-responsive focus. \systemName dynamically adjusts the cluster capacity and scale of batch jobs in response to variations in the grid signal, workload demand, and workload scalability.
    
    \item \textbf{Separate Provisioning from Scheduling:} \systemName views cluster resource management as two distinct tasks: provisioning and scheduling. The provisioning policy $\phi(\cdot)$ determines how many servers to acquire from the cloud for the entire cluster, while the scheduling policy $\psi(\cdot)$ determines which batch jobs to run on the available servers and how many servers to allocate to each. This separation of provisioning and scheduling is similar to that in other frameworks, such as Mesos~\cite{mesos}. Importantly, this separation is at the level of \emph{execution}, and \systemName still learns the two policies \emph{jointly} from the oracle, so provisioning accounts for how jobs actually scale.

    \item \textbf{Data-Driven Learning:}
    Lastly, inspired by \emph{hindsight learning}~\cite{Sinclair2023:Hindsight}, \systemName takes a \emph{data-driven} approach that simulates an offline oracle policy over past job arrivals to determine how an optimal algorithm (with full knowledge of job characteristics, the grid signal, and future arrivals) would schedule jobs cost-efficiently, and distills those decisions into the provisioning and scheduling parameters it uses at runtime. 
    This works because the grid signal and arrivals \emph{recur} with diurnal and weekly regularity, so a situation \systemName faces at runtime usually resembles one in the recent past, and the oracle's decision for it remains effective. In addition, we handle workload changes by relearning these parameters over a rolling window, adapting to distribution shift. Nonetheless, \systemName includes a basic fallback policy that serves as a safeguard when it cannot mimic the optimal policy's behavior, ensuring reasonable performance during sudden changes.
\end{enumerate}

Together, these principles address the three challenges from \autoref{sec:intro}. \systemName's data-driven learning replays an offline oracle that jointly optimizes capacity and scheduling over a past window, thereby provisioning capacity in light of how jobs actually scale (Challenge~1) and how earlier decisions reshape the future queue (Challenge~2). Elastic scaling provides the knobs to implement these decisions at both the cluster and individual job levels, while separating provisioning from scheduling allows the two to operate at their natural timescales (Challenge~3).

\autoref{fig:system_overview} depicts \systemName's architecture and shows how the design principles above are instantiated through the learning and execution phases. In doing so, \systemName employs a two-step approach: 1) a Learning Phase, where \systemName employs data-driven learning over the most recent cluster execution traces and captures the key decisions at different runtime states, and 2) an Execution phase, where \systemName utilizes such knowledge to enhance its provisioning and scheduling, as we detail below.

\subsection{Hindsight Learning Phase}\label{sec:des_learn}

\systemName's hindsight learning phase (see ~\autoref{fig:system_overview}) employs data-driven learning on recent cluster execution logs by replaying them to an offline oracle algorithm and learning from its provisioning and scheduling decisions. The process involves periodically (\eg daily) tracking the job arrival logs and grid signal traces over a window of length $T$ and replaying those traces to a simulated oracle algorithm. Note that the oracle algorithm cannot be employed in practice, as it requires full knowledge of the job arrival sequence, job characteristics, and grid signal variations. However, since the learning phase operates on historical traces, the entire arrival trace, job characteristics, and grid signal traces are known over the window $T$, allowing simulation of an oracle's decisions over this past window.

These oracle decisions can be viewed as mappings from the overall system state at each time $t$ to the cluster capacity used for that state and the scheduling behavior in that case.
As an example, consider a simplified system state described by two parameters: a grid signal value at time $t$, denoted by $c_t$, and $N_t$, which captures the number of jobs (running and queued) in each job queue. 
In this case, the tuple mapping $(c_t, N_t) \mapsto (m_t,\delta)$ denotes the cluster capacity $m_t$ that was provided by the oracle for that system state, while $\delta$ denotes the lowest marginal throughput across all scheduled jobs, indicating that no jobs with elastic scaling curves below this threshold were chosen for execution. 
The tuple mappings from the oracle's simulated decisions at each time step $t$ are then stored in a knowledge base that is later consumed during the execution phase.

\RestyleAlgo{ruled}

\begin{algorithm}[t]
    \caption{\texttt{Flex\_Oracle\_Algorithm()}}
    \label{alg:offline}
    \KwIn{Jobs $\J$, Max Resources $M$, Grid signal $c$.}
    \KwOut{Schedule $S$}
    \textbf{Initialization:} $S \gets \{s_1, ...,s_{|\J|}\}$ and $L \gets [ \ ]$\;%
    \For {$j\in \J$}{%
        \For{$t \in [a_j, e_j]$}{%
            \For{$k \in [k_j^{\min}, k_j^{\max}]$}{%
                $L$.append($j, t,k, \rho_{j,t}^{(k)}, e_j$)\;
            }
        }
    }
    $L \gets$ Sort($L$) ; \tcp{w.r.t. $\rho$ then $e_j$ in ascending order.}\label{algline:sorting_offline}
    \While{$|L|>0$}{%
        $j,t,k,*,* \gets L.pop()$; \tcp{next lowest $\rho$}
        \If{$\sum_{j'\in \J \setminus {j}}s_{j',t} + k - s_{j,t} >M$}{%
            continue; \tcp{Cannot scale the current job.}
        }
        \If{$progress(s_j)< 100\%$}{%
        \tcp{Job not done.}
            $s_{j,t} = k$; \tcp{allocate $k$ resources to job $j$ at time $t$.}
        }
    }
    \For{$s_j \in S$}{
        \If{$progress(s_j)< 100\%$}{%
            \Return None\ \tcp{Non Feasible}
        }
    }
    \Return $S$
\end{algorithm}
\noindent\emph{\textbf{\systemOracle.}} The key challenge in designing the \systemOracle lies in developing a scalable approach that can periodically augment its knowledge base with the most recent load and grid signal patterns. While directly using the ILP in \autoref{sec:prob}  can yield a functional oracle, it is not a scalable approach.

In addressing this challenge, we propose a polynomial-time greedy optimal algorithm that generates an execution schedule and the time-varying cluster capacity required to execute it over a past window. 
The algorithm takes historical job traces for a time horizon $\T$ (\eg a week) containing $\J$ jobs. 
Each job is characterized by an arrival time $a_j$, job length $l_j$, allowed delay $d_j$ (based on the selected job queue), and a marginal throughput profile $\Delta_j$, as explained in ~\autoref{sec:prob}. 
The algorithm then creates the execution schedule that optimizes the cluster's total cost under the grid signal by constructing a schedule $s$ for each job that respects its arrival and delay constraints and the maximum cluster capacity $M$.

\autoref{alg:offline} uses a greedy approach to elastically scale and schedule jobs.
The key insight is that for any given signal value, doing more work per unit of energy (i.e., greedily choosing jobs with lower \emph{cost-per-benefit}) yields better cost efficiency~\cite{Subramanya2023:Sia, Hanafy2023:CarbonScaler}. The algorithm starts by computing the cost-per-benefit across time slots and scales, $\rho_{j,t}^{(k)} = \frac{c_t}{\Delta_j(k)}$, where for a job $j, t\in[a_j, e_j]$ and $k \in [k_j^{\min}, k_j^{\max}]$.
Then, it sorts the list (\autoref{algline:sorting_offline}) in ascending order of cost-per-benefit, using deadlines as a tie-breaking rule. For example, when two jobs have the same scalability (\eg both at scale one), the job with the earliest deadline is prioritized. 
The algorithm then iterates over the list, greedily assigning resources to jobs while respecting the maximum capacity, using the $progress(\cdot)$ function. Finally, it verifies that all jobs have completed. Otherwise, it marks the schedule as infeasible for the specified cluster capacity and job delays. In this case, we repeat the algorithm while extending the deadline for jobs that were not finished.

\noindent\emph{\textbf{Runtime Complexity and Optimality.}}
~\autoref{alg:offline} runs in polynomial time. To schedule a trace with $N$ jobs and $K$ scaling states on a cluster of $M$ servers, the complexity of computing the cost-per-benefit ratio (Lines 2-5) is $\mathcal{O}(N \cdot K \cdot T)$, marginal throughput list sorting (Line 6) is $\mathcal{O}(N \cdot K \cdot T \cdot \log(N \cdot K \cdot T))$, and iterating over possible resource allocations (Lines 7-13) is $\mathcal{O}(N \cdot K \cdot T)$, and finally, the time complexity of job completion validation (Lines 14-16) is $\mathcal{O}(N)$. The total time complexity is $\mathcal{O}(N \cdot K \cdot T + N \cdot K \cdot T \cdot \log(N \cdot K \cdot T) + N)$ $\simeq \mathcal{O}(N \cdot K \cdot T \cdot \log(N \cdot K \cdot T))$.

\begin{theorem}[\autoref{alg:offline} Optimality]
\label{theorem:CFO}
Consider elastic jobs with $k_j^{\min}=1$, non-increasing marginal throughput profiles, homogeneous resources, a non-negative grid signal, negligible scaling overheads, and an instance that admits a feasible schedule. \autoref{alg:offline} returns an optimal schedule that minimizes the cluster's total cost under the grid signal.
\end{theorem}

\begin{proof}
We detail our proof of \autoref{theorem:CFO} in \autoref{app:proof}.
\end{proof}

\begin{remark}[Relationship to marginal resource allocation problems (MRA)]
\label{rem:MRA}
\citet{greedy_optimal} showed that for separable resource allocation with monotone marginal \emph{costs}, the greedy algorithm that repeatedly takes the smallest current marginal cost is optimal.
Here, equivalently, we minimize total \emph{cost} subject to meeting \emph{progress} (work) requirements with non-increasing marginal throughput (i.e., concave) within-slot benefit.
This reduces to picking the smallest \emph{cost-per-benefit} ratio as in \autoref{alg:offline}, where concavity ensures the cost-per-benefit ratios on any fixed $(j,t)$ are non-decreasing, as the hypothesis of MRA.
\end{remark}

\begin{table}[t]
\captionof{table}{State and output decisions collected by \systemOracle.}
\label{tab:features}
\resizebox{\linewidth}{!}{%
\begin{tabular}{l|c}
\toprule
\textbf{State} & \textbf{Explanation}\\
\midrule
$c_t$ & Grid signal (\eg carbon intensity in \intensityunit). \\
$c_t^G$ & Gradient of the signal at $t$. \\
$c_t^R$ & Rank of slot $t$ vs.\ day-ahead signal. \\
$\{N_t^Q\}$ & Number of jobs (paused + running) per queue.\\
$\eta$ &  Average elasticity across all jobs in the system.\\

\bottomrule \toprule
\textbf{Decision} & \textbf{Explanation}\\
\midrule
 $m_t$ & The cluster capacity at time $t$. \\
 $\delta$ & The minimum marginal throughput used.\\
\bottomrule
\end{tabular}%
}
\end{table}

\noindent\emph{\textbf{Retaining Oracle decisions.}}
~The output of the offline oracle algorithm is visually depicted in ~\autoref{fig:idea} and can be viewed as (i) the provisioned cluster capacity at time $t$, which varies over time, and (ii) how these servers are assigned among jobs, where jobs are not scaled until all jobs are assigned a single resource. Given these decisions over the time window $T$, \systemName learns the provisioning policy as a mapping from the current system state at each time step $t$ to the cluster size chosen at that step, i.e., a function $STATE \mapsto m_t$. As discussed, the simplest representation of the current system state is the state of the job queues (\eg the number of queued and running jobs in each queue) and the current grid signal values $(c_t, N_t)$. 

In practice, our approach uses several other parameters to fully capture the current state, as shown in  Table~\ref{tab:features}. These include the signal gradient (whether the grid signal is increasing or decreasing), the day-ahead ranking $c_t^R$ (how favorable the current slot is compared to the day-ahead signal forecast),
the number of jobs per queue $N_t^Q$, and $\eta$ the mean elasticity of all jobs in the system, where a job's elasticity captures how much its throughput grows with additional resources (\ie its marginals $\Delta_j$). Similarly, the scheduling threshold is defined as a mapping $STATE \mapsto \delta$, indicating that the oracle scheduling policy schedules only jobs with marginal throughput above the threshold. The provisioning and scheduling policy decisions per-state mappings are then stored in a knowledge base that is later used during the execution phase. Finally, older mappings from the knowledge base are aged out over a rolling window to adapt to seasonal variations in the grid signal and any changes in the workload distribution over time.

\subsection{Run-time Provisioning and Scheduling}\label{sec:des_ex}

\systemName's execution phase implements runtime provisioning and scheduling algorithms that optimize cost under the grid signal while respecting the queue-specific delays. \systemName runs these two decisions at \emph{two timescales}, updating capacity provisioning once per hourly slot to track the grid signal and rescheduling jobs at a finer scale, on the order of minutes or whenever a job arrives, finishes, or capacity changes. These algorithms use knowledge derived from the oracle's decisions during the offline learning phase, where we execute the \systemName Oracle on historical data, to make real-time decisions.\footnote{Our experiments in \S \ref{sec:eval} show that even when some of these optimality conditions may not hold in practice (\eg variations of power consumption across homogeneous nodes, or switching costs when scaling), the oracle provides significant (but not necessarily optimal) insights, allowing our learning-based approach to be effective at runtime.} 
At runtime, users submit their jobs to the cluster by selecting a job queue, which sets the delay budget $d_j$. \systemName treats $d_j$ as a soft, best-effort target and tracks each job's accrued delay relative to it, since a deadline is only meaningful when the job length is known, which is typically unknown at submission time~\cite{Hanafy2024:GoingGreen, Kuchnik2019ThisIW}.

\RestyleAlgo{ruled}
\begin{algorithm}[t]
    \caption{\texttt{Flex\_Decision\_Matching()}}
    \label{alg:prov_policy}
    \KwIn{$STATE$,  $\mathbb{k}$, Delay violations $v$, Distance threshold $\tau$, Violation tolerance $\epsilon$}
    \KwOut{Cluster capacity $m_t$ and scheduling threshold $\delta$}
    $\Re \gets $match($STATE$, $\mathbb{k}$)\\
    \If{Distances($\Re)>\tau$ AND $v>\epsilon $}{
        \Return M, 0 \tcp*{unfamiliar state and delays}
    }
    \ElseIf{$v>\epsilon $}{
        \Return Max($\Re.m_t$), Min($\Re.\delta$) \tcp*{delays: recover}
    }
    \Return Mean($\Re.m_t$), Mean($\Re.\delta$) \tcp*{mimic matches}
\end{algorithm}

\noindent\emph{\textbf{Runtime Decision Matching.}}
At each time slot $t$, \systemName computes the current system state from the attributes in \autoref{tab:features} and queries the knowledge base for the top-$\mathbb{k}$ closest past states (\eg by Euclidean distance).
\autoref{alg:prov_policy} then reuses the oracle's stored decisions for those matches, returning both the cluster capacity $m_t$ (applied by the provisioning policy $\phi$) and the scheduling threshold $\delta$ (applied by the scheduling policy $\psi$). 
Generally, \autoref{alg:prov_policy} returns the \emph{mean} of the matches' $m_t$ and $\delta$, mimicking what the oracle did in similar situations.
However, \autoref{alg:prov_policy} incorporates a safeguard mechanism that overrides the oracle’s decisions when jobs experience delays exceeding their target thresholds.
If the recent delay-violation rate $v$ exceeds a tolerance $\epsilon$ but the state still resembles known cases (distance $\le \tau$), \systemName raises capacity to the \emph{maximum} of the matches and lowers the threshold to their \emph{minimum}, admitting more work to catch up while still trusting the learned policy. If delays persist and the state is unfamiliar (distance $> \tau$), as under an abrupt workload or grid signal spikes, the learned policy is no longer reliable, so \systemName falls back to \emph{signal-agnostic} execution, provisioning the maximum capacity $M$ and setting $\delta=0$ to admit every job until the backlog clears.

\RestyleAlgo{ruled}
\begin{algorithm}[t]
    \caption{\texttt{Flex\_Scheduling()}}
    \label{alg:scheduling_policy}
    \KwIn{Current time $t$, Current jobs $N_t$, Available resources $m_t$, Marginal throughput threshold $\delta$.}
    \KwOut{Resource allocation $S_t$}
    \textbf{Initialization:} $S \gets \{\}$ and $L \gets [ \ ]$\;%
    \For {$j\in N_t$}{%
        \For{$k \in [k_j^{\min}, k_j^{\max}]$}{%
            \If{$\Delta_j(k)>\delta$}{
            $L$.append($j, k, \Delta_j(k), a_j + d_j - t$)\;
            }
        }
    }
    $L \gets$ Sort($L$) ; \tcp{w.r.t. $\Delta_j(k)$ then $a_j + d_j - t$}\label{algline:sortingonline}
    \While{$|L|>0$ and $\sum_{j\in N_t}S_{j}[t] <m_t$}{%
        $j,k,*,* \gets L.pop()$; \tcp{next highest $\Delta_j(k)$}
        $S[j] = k$; \tcp{increase allocation of $j$}
    }
    \Return $S$
\end{algorithm}

\noindent\emph{\textbf{Workload Scheduling.}} After ``right-sizing'' the cluster at the start of the time slot $t$,  the scheduling algorithm in \autoref{alg:scheduling_policy} then uses the marginal-throughput threshold $\delta$ from \autoref{alg:prov_policy} to decide which jobs to schedule and how much to allocate to each scheduled job.
To do so, the algorithm iterates over the current jobs and selects all job scales with 
a \emph{marginal throughput threshold} $\delta$, that \systemName learns from the oracle, where a job is considered to run at a scale of $k$ servers only if $\Delta_j(k)> \delta$.
Then, it sorts assignments according to their marginal throughput and available delays (\autoref{algline:sortingonline}).
Finally, the algorithm iterates over the list of schedules, choosing jobs until the current capacity $m_t$ is filled.

\noindent\emph{Note that} (i) This approach does not directly consider the cost-per-benefit ratio $\rho_{j,t}^{(k)}=c_t/\Delta_j(k)$, since in a single time slot the grid cost is fixed. (ii) 
Under a capacity limit $m_t$, all jobs are assigned $k_j^{\min}$ before any scaling, since $\Delta_j(k_j^{\min})=1 \ \forall j$. 
Thus, no job is starved, and admission follows the remaining delay budget $a_j+d_j-t$ (earliest soft-deadline first, reducing to first-come-first-served within a queue). (iii) In practice, \autoref{alg:scheduling_policy} does not require exact elasticity profiles but rather the ability to sort them based on the impact of adding another resource from more to less scalable.

%% file: sections/4-implementation.tex
We implement \systemName using AWS ParallelCluster~\cite{aws-pc}, a cluster management tool that deploys and manages high-performance computing (HPC) Slurm~\cite{2003slurm} clusters in the cloud.
AWS ParallelCluster uses EC2 instances that span various hardware configurations, networks, and accelerators. We implement \systemName using PySlurm~\cite{pyslurm} as a Slurm interface to submit workloads to the cluster according to our underlying provisioning and scheduling policies.
Our prototype, available at (\url{https://github.com/umassos/Flex}), has the following components:

\noindent\emph{\textbf{Data-Driven Learning.}}
We implement our \systemOracle in a simulation environment using Python, that utilizes historical workload traces, scheduling profiles, configurations, and historical grid signal data (\eg carbon intensity) to compute the cost-optimal schedule and the historical ($STATE \mapsto m_t, \delta$) mappings. 

\noindent\emph{\textbf{Runtime Decision Matching:}}
At each slot, which we set to 1 hour, \systemName computes the number of provisioned servers and scheduling configurations based on the decisions made by the offline oracle. 
Our implementation finds similar states using KNN from the Scikit-learn library~\cite{scikit-learn}, where we standardize the state features with Scikit-learn's \texttt{StandardScaler} before computing Euclidean distances. We represent the historical states in a KD-Tree for fast access, select the nearest five instances $\mathbb{k}=5$, and combine them as detailed in ~\autoref{alg:prov_policy}.

\noindent\emph{\textbf{Elastic Scaling and Scheduling.}}
After computing available capacity, 
\systemName schedules workloads based on their marginal capacity, as detailed in \autoref{sec:des_ex}.  
When the capacity or number of queued jobs changes, or every $\Delta t$ (\eg, 5 minutes), \systemName computes each job's resource assignment and scale using \autoref{alg:scheduling_policy}, which can be efficiently implemented with a binary tree. To run a job, \systemName submits it as a Slurm job to the appropriate job queue using PySlurm, allocating whole, dedicated nodes per job to avoid co-location interference.  
Finally, to scale jobs, \systemName uses the \texttt{scancel} command, which signals the job to checkpoint its state and to submit a new job at the new scale, resuming progress.

\noindent\emph{\textbf{Energy and Cost Monitoring.}}
As stated earlier, \systemName focuses on energy-based operational costs.
We compute the cluster's operational cost under the grid signal at time $t$, denoted as $\mathcal{C}_t$ (its carbon emissions when the signal is carbon intensity), as follows:

{
\begin{align} 
\label{eq:con_scaling}
    \mathcal{C}_t = \sum_{j}^{N_t} E_{js} \times c_t \\
    E_{js} = E_{js}^R + E_{js}^{net} \\
    E_{js}^{net} = \eta_{net}\times Mem_{js}
\end{align}
}
where $E_{js}$ represents the energy consumption by job $j$ at scale $s$, which consists of compute and network components, denoted as $E_{js}^R$ and $E_{js}^{net}$, respectively.\footnote{A job can take multiple scales or utilize part of the time slot, details we omitted in this model for clarity.} We attribute energy to individual jobs, as resources are not shared and the cluster's total is then the sum over concurrently running jobs.
$E_{js}^R$ takes into account the number of compute resources (\eg CPU cores) and can be augmented to include energy consumption for memory, base power, and PUE. However, due to the challenges of accurately assessing per-tenant power consumption in data centers, our CPU experiments use clusters where these numbers are often unknown, and we assume a fixed value per resource, a common approach in carbon accounting~\cite{Gsteiger2024:Caribou, teads2021carbon, Lannelongue2021:GreenAlgorithms}. 
Further, our results emphasize normalized savings, making absolute energy or cost values less significant. In contrast, for GPU experiments, we use \texttt{nvidia-smi}, which allows us to measure energy usage for each GPU and aggregate across the scale used by this job.
To account for the network cost $E_{js}^{net}$, we utilize the network energy efficiency ($\eta_{net}$) measured in $W/Gbps$, and $Mem_{js}$ is the amount of data transferred by the job $j$ at scale $s$. Since $\eta_{net}$ widely varies in prior work, by up to three orders of magnitude (\eg due to network type and topology)~\cite{Tabaeiaghdaei2023:GreenRouting, Jacob:GrowOld, Heddeghem2012:optical_networks}, our experiments utilize a value of 0.1$W/Gbps$.
Lastly, we compute $Mem_{js}$ based on the communication paradigm used for the workload as per its implementation and memory requirements (\eg distributed training~\cite{pytorch} uses \texttt{Ring-Allreduce}).

\noindent\emph{\textbf{Simulation Environment.}}
Lastly, we integrate the proposed scheduling policies and the baselines into a simulation environment, denoted as \systemNameSim, which enables year-long evaluation. The simulator replays the grid signal and workload traces and computes the schedule and its cost as per the scheduling policies.

%% file: sections/5-eval.tex
This section evaluates \systemName's grid-responsive provisioning and scheduling policies under time-varying grid signals.
We evaluate \systemName within our real-world CPU and GPU prototype deployments on AWS ParallelCluster~\cite{aws-pc}, reporting savings against the \emph{carbon intensity} signal alongside job delay. 
We then augment our evaluation with year-long \systemNameSim simulations that sweep configurations, workloads, and grid signals.
Lastly, we present a sensitivity analysis and system overheads.

\begin{figure}[t]
    \centering
    \includegraphics[width=1\linewidth]{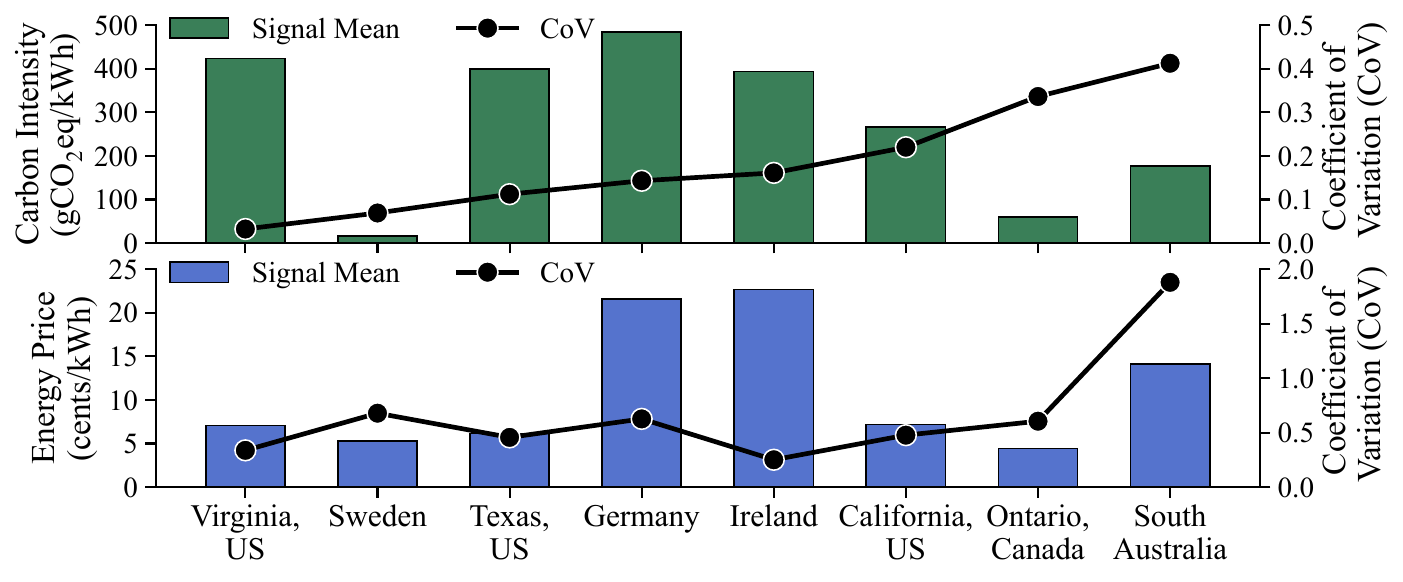}
    \caption{Diversity of selected carbon and price traces.}
    \Description{Diversity of selected carbon and price traces.}
    \label{fig:carbon_trace_ccs}
\end{figure}

\subsection{Experimental Setup}\label{sec:eval_setup}
\noindent\emph{\textbf{Workload Traces.}} 
Our experiments use three workload traces with different arrival patterns and job lengths: a month-long \emph{Azure} trace~\cite{azure-data-paper}, a two-month \emph{Alibaba-PAI} trace~\cite{weng2022mlaas}, and a year-long \emph{SURF Lisa-HPC} trace~\cite{Chu2024:SURF}.
In our experiments, we focused on hour+ workloads, as shorter jobs contribute little to total compute time and are usually not delay-tolerant.
Then we sample these traces to create \emph{historical} and \emph{evaluation} traces. We split the \emph{historical} trace into two week-long traces and used them for the learning phase, while the \emph{evaluation} trace is a week-long trace used to evaluate our proposed approaches. We sample these traces from different parts of the trace. 
For instance, we use the first two weeks of the \emph{Azure} trace to sample the \emph{historical} trace and the third week to sample the \emph{evaluation} trace. In contrast, for longer traces such as \emph{Alibaba-PAI}, which span two months, we used the first 7 weeks for learning and the 8th week for evaluation.
Finally, our elastic workloads (see \autoref{app:elastic_workloads}) are MPI jobs (\eg $N$-body simulations) run on CPU servers and PyTorch training jobs (\eg ViT) run on GPU servers, profiled at scales up to 16 CPU and 8 GPU servers, respectively. Except for the GPU-cluster experiment (\autoref{fig:gpu_cluster}), all experiments use the CPU cluster and its MPI workloads. Unless otherwise stated, we assign jobs with a random scaling profile (our ``Mix'' scenario in \autoref{sec:eval_ccs}).

\begin{figure*}[t]
    \centering
    \hfil
    \begin{subfigure}[b]{0.49\textwidth}
        \centering
        \includegraphics[width=\linewidth]{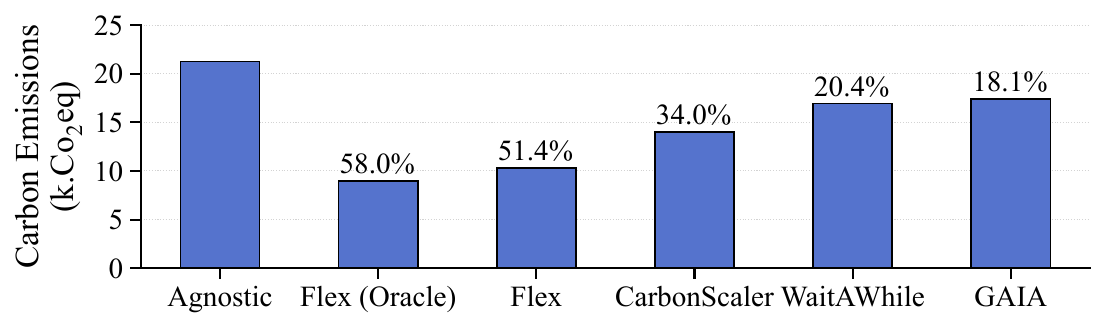}
        \caption{Carbon Emissions and Savings (on-top)}
        \label{fig:perf_cpu_cluster_carbon}%
    \end{subfigure}
    \hfil
    \begin{subfigure}[b]{0.49\textwidth}
        \centering
        \includegraphics[width=\linewidth]{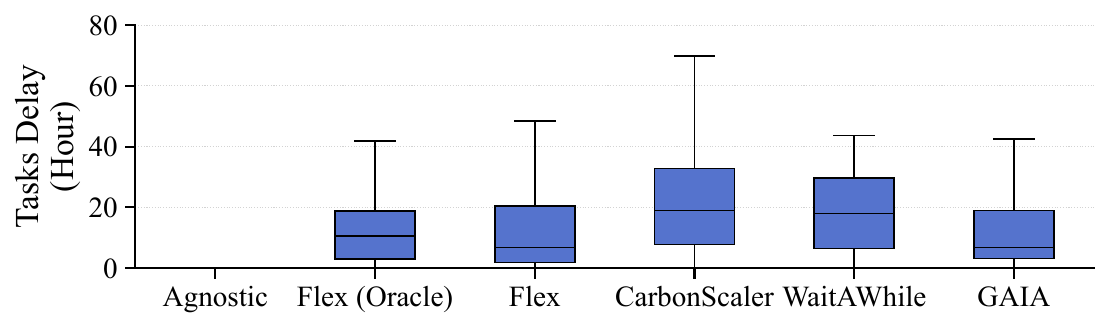}
        \caption{Delay}
        \label{fig:perf_cpu_cluster_delay}%
    \end{subfigure}
    \hfil
    \hfil
   \caption{Carbon emissions (a) and delay (b) across scheduling approaches for the CPU cluster under carbon intensity.}%
   \Description{Carbon emissions (a) and delay (b) across scheduling approaches for the CPU cluster under carbon intensity.}
    \label{fig:perf_cpu_cluster}
\end{figure*}

\noindent\emph{\textbf{Grid-Signal Traces.}} 
We used hourly carbon intensity and cost traces for 8 geographical regions from December 2021 to December 2022. Carbon intensity traces are from Electricity Maps~\cite{electricity-map}, while cost traces are from different sources~\cite{germany_prices, ireland_prices, lcg_energyonline_industry_data, australia_prices, pjm_prices, ontario_prices, swedent_prices}. Unlike carbon intensity, the cost traces include zero and negative prices (\eg during renewable surplus), which are undefined or degenerate for many scheduling algorithms, so we floor prices to zero.
~\autoref{fig:carbon_trace_ccs} shows the mean carbon intensity and price of energy, and their daily variability, measured by the daily Coefficient of Variation (CoV) throughout this period. 
As shown, the selected regions span a range of conditions, from high-cost regions with limited fluctuation (\eg Virginia, US and Ireland) to regions with high fluctuation (\eg Ontario, Canada and South Australia).

\noindent\emph{\textbf{Baselines.}} We compare \systemName against several signal-aware scheduling baselines.\footnote{While these baselines were proposed for carbon intensity, we utilize them to optimize whichever grid signal is active.}
For fairness, we adapt these baselines to schedule workloads within a variable capacity cluster, and we assume that all baselines have access to historical traces and can use the mean job length for computing the schedule:
\begin{enumerate}[leftmargin=*]
\item \textbf{\agnostic}: This policy represents the status quo, where jobs are scheduled FCFS without elastic scheduling, which we use as a baseline to compute the savings for all policies.

\item \textbf{\GAIA}~\cite{Hanafy2024:GoingGreen}: We utilize GAIA's Lowest-Window Policy, which schedules jobs in a non-elastic manner by selecting the best start time based on the mean job length within a time window $w$ to minimize costs. We augment the policy with resource limits and use FCFS to resolve contentions.

\item \textbf{\WaitAWhile}~\cite{Wiesner2021:WaitAwhile}: 
We implement the threshold version of the 
\WaitAWhile policy, which operates the job in a suspend-resume fashion according to the grid signal. The threshold is determined by the 30th percentile of grid signal predictions for the next 24 hours. To respect the delay target, the job runs to completion once the permitted delay is exhausted, and we use FCFS to resolve contentions.

\item \textbf{\CarbonScaler}~\cite{Hanafy2023:CarbonScaler}: We adapted the \CarbonScaler algorithm to run on a multi-job cluster, where the schedule is computed based on historical job lengths. In addition, to respect the cluster-wide capacity, we prioritize scaling jobs with higher marginal throughput. Lastly, when the job surpasses its delay target, it runs until completion.

\item \textbf{\GREEN}~\cite{Xu2025:Green}: 
We implement \GREEN, a cluster-level extension of GREEN that uses elastic scheduling and job reordering to reduce total cost and makespan, setting the allowed degradation $\gamma$ to 0.9 as the paper recommends.
To make it comparable to our delay-aware policies, \GREEN implements a per-job delay budget, in which jobs may wait while slack remains but are prioritized for execution once their target delay is exhausted. 

\item \textbf{\systemOracle}: Finally, we added the offline oracle as a baseline that implements ~\autoref{alg:offline} and assumes full knowledge of the grid signal and job length.

\end{enumerate}

\noindent\emph{\textbf{Deployment.}} We deployed \systemName in AWS and evaluated it on a CPU and GPU cluster. In our CPU cluster, we utilize 150 \texttt{c8g.medium} VMs. For the GPU cluster, our resource quota only allowed for 15 \texttt{g6.xlarge} GPU instances, so we limited the sampling to ensure similar utilization for this cluster size. To simulate the behavior of \systemOracle, specifically the learning phase, we replay the available historical trace with different start times, a step that helps improve the performance of \systemName. Finally, we augment our evaluation with year-long assessments using \systemNameSim to evaluate many different scenarios and settings. 
Note that, in all experiments, unless otherwise stated, we utilize the carbon intensity trace of South Australia and sampled workloads randomly. We focus on South Australia because its carbon and price signals are among the most variable in our set (see \autoref{fig:carbon_trace_ccs}), establishing the effectiveness of our data-driven approach in a complex scenario, while reporting its behavior across grids in \autoref{sec:eval_location_signal}. 
We also utilize a cluster capacity, representing an average of 50\%, matching the utilization of real-world clusters~\cite{Tirmazi2020:Borg, Shehabi2024:USDCReport, Shehabi2016:USDCReport}. However, we later evaluate tighter regimes with different cluster sizes (See~\autoref{fig:cluster_size}).
Finally, we consider three length-based queues with $d=6hrs, 24hrs$, and $48hrs$ for short ($l\leq2$hrs), medium ($2< l\leq12$hrs) and long ($l>12$hrs) jobs. Nonetheless, we revisit these assumptions in later sections.%

\subsection{Prototype Evaluation}
We first evaluate \systemName on real CPU and GPU clusters under \emph{carbon intensity} signal, utilizing carbon emissions as our cost metric.

\noindent\emph{\textbf{CPU Cluster.}}
First, we evaluate \systemName on our prototype AWS CPU cluster of \texttt{c8g.medium} instances, with $M=150$. ~\autoref{fig:perf_cpu_cluster} shows the total carbon emissions of our cluster under different scheduling baselines. 
As shown, \systemName can reduce carbon emissions by 51.4\% (only 6.6\% away from \systemOracle) and achieve 17.4\%, 31\%, and 33.3\% higher savings than \CarbonScaler, \WaitAWhile, and \GAIA, respectively. In addition, the results show that approaches that use scaling (\eg \systemName and \CarbonScaler) offer higher savings as they can better utilize variations in carbon intensity, while approaches that use suspend-resume scheduling (\eg \WaitAWhile) perform better than those that do not consider preemption. 

\autoref{fig:perf_cpu_cluster_delay} shows the delay experienced across baselines, where the \agnostic baseline exhibits no waiting and  \systemOracle respects all delay targets.
As shown, policies typically respect the delay, where all approaches are configured to run to completion once the allowed delay period is over. The highest delays, however, are exhibited by scale-based approaches (\eg  \systemName and \CarbonScaler) as the use of provisioning in \systemName may limit the cluster capacity and limit how jobs are scheduled, leading to an average delay of 17.5 hours. 
In addition, \CarbonScaler may under-predict job length and delay jobs beyond their allowed delay, requiring them to run to completion, leading to an average delay of 22.3 hours. 
\systemName often incurs a lower average delay than \systemOracle, which schedules aggressively toward the maximum delay.

\begin{figure}[t]
    \centering
    \includegraphics[width=\linewidth]{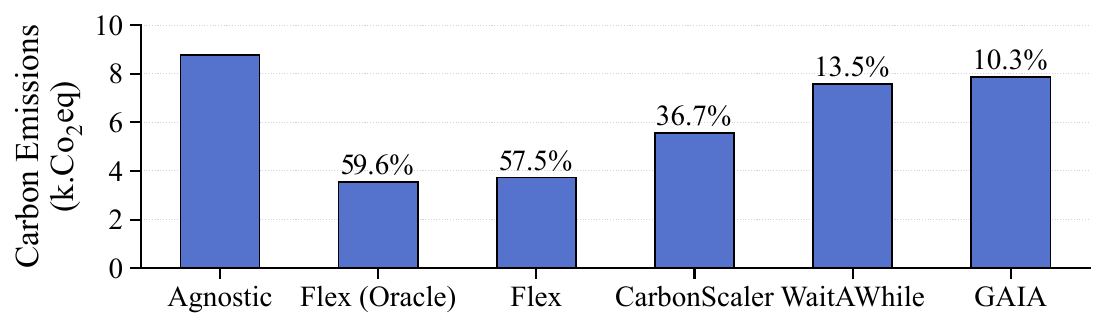}
    \caption{Carbon emissions and savings (on-top) across scheduling approaches in a GPU cluster.}
    \Description{Carbon emissions and savings (on-top) across scheduling approaches in a GPU cluster.}
    \label{fig:gpu_cluster}
\end{figure}
\noindent\emph{\textbf{GPU Cluster.}}
~\autoref{fig:gpu_cluster} shows the carbon emissions and savings using our prototype evaluation on a 15 \texttt{g6.xlarge} GPU cluster on AWS. As shown, \systemName significantly reduces carbon emissions, achieving 57.5\% savings, within 2.1\% of \systemOracle. As in the CPU cluster, \systemName reduces carbon emissions by 20.8\%, 44\%, and 47.2\% compared to \CarbonScaler, \WaitAWhile, and \GAIA. 
Interestingly, in our GPU cluster, where applications have heterogeneous power draw, scaling approaches outperform temporal shifting. 
This occurs because scaling approaches prioritize workloads with higher marginal throughput during low-cost periods (i.e., low communication per unit compute), which typically have higher GPU utilization and consume more power, and shifting them further enhances savings.

\noindent\emph{\textbf{Key Takeaways.} 
\systemName closely follows the \systemOracle, increasing carbon savings by up to 57.5\% and 20.8\% compared to \agnostic and \CarbonScaler, respectively.
}

\subsection{Effect of Configurations}
This section utilizes the \emph{carbon intensity} signal (with carbon emissions as the cost) to demonstrate how cluster configurations (\eg cluster capacity and delay) affect savings and \systemName's ability to adapt its decisions based on these configurations.

\begin{figure}[t]
    \centering
    \includegraphics[width=\linewidth]{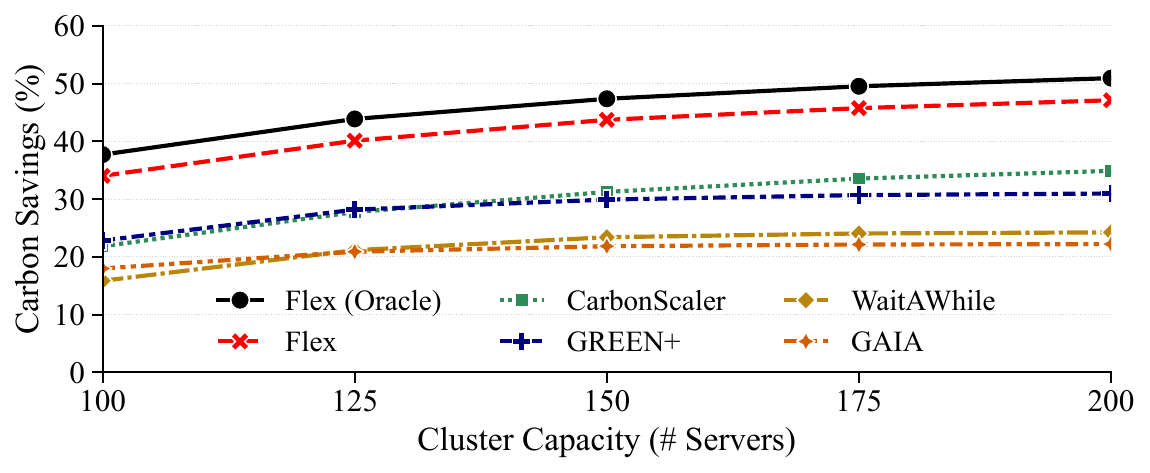}
    \caption{Impact of cluster capacity on carbon savings.}
    \Description{Impact of the maximum cluster capacity on the carbon savings.}
    \label{fig:cluster_size}
\end{figure}

\noindent\emph{\textbf{Effect of Cluster Capacity.}}
The maximum cluster capacity represents the headroom available to stack workloads during low-cost periods. 
\autoref{fig:cluster_size} shows the effect of this headroom, the maximum cluster capacity $M$, for $M=100$, 150, and 200, corresponding to average utilizations of $\sim$75\%, $\sim$50\%, and $\sim$37\%.
As shown, \systemName closely follows \systemOracle across all cluster capacities, achieving within $\sim$3.7\% of the \systemOracle. In addition, \systemName outperforms other approaches, such as \CarbonScaler and \GREEN
with up to 12.5\% and 16.1\% savings, respectively.
Moreover, the results show that elastic scheduling better utilizes the added capacity to further reduce emissions, whereas temporal-shifting-only approaches, when scaling the cluster from $M=100$ to $M=200$, increase the savings by only 8.4\%.
In addition, increasing the cluster size shows diminishing returns, where raising $M$ from 100 to 200 improves savings by only 13.2\% and 13\% for \systemOracle and \systemName, respectively.
Lastly, although larger clusters increase energy consumption because applications run at lower marginal throughput~\cite{Hanafy2023:War}, the energy overhead remained below 3.2\% across all methods and sizes.

\noindent\emph{\textbf{Effect of Delay.}}
The delay reflects the scheduling flexibility of workloads, a central enabler of signal-aware scheduling. From the grid's perspective, this delay budget determines how much flexibility the cluster can supply, with a larger budget allowing \systemName to defer and reshape more load in exchange for longer completion times. Throughout, we treat each job's allowed delay as a soft, best-effort target rather than a hard deadline, so jobs may occasionally run beyond it, and \systemName aims to keep such violations to a minimum.
\autoref{fig:delay} shows the impact of increasing the delay on the carbon savings and waiting time for \systemName and other baselines, assuming all jobs share the same queue and delay target. We change the allowed delay per job from 0hrs (using only elasticity) to $36$hrs. 
\autoref{fig:delay_carbon} shows that increasing the allowed delay from $d=0$hrs to $d=36$hrs increases carbon savings by 18.3\% and 13.8\% (total savings of 45.7\% and 42.1\%) for \systemOracle and \systemName, respectively. 
Notably, the figure shows that \systemName follows the \systemOracle, where it achieves carbon savings within 3.6\% of \systemOracle's savings. 
The figure also shows how other baselines behave under varying degrees of temporal flexibility. 
For instance, approaches such as \WaitAWhile, which rely solely on temporal shifting, yield no carbon savings when $d=0$ and reduce savings by only 19.7\% when $d=36$. 

In contrast, approaches that can leverage elasticity (\eg \CarbonScaler) can achieve much higher savings than non-elastic baselines. 
The figure also highlights that \GREEN substantially benefits from additional waiting slack. When no delay is allowed, \GREEN cannot intentionally defer jobs, as the waiting budget expires at arrival, so jobs become immediately eligible for execution and the policy achieves no carbon savings. However, as the allowed delay increases, \GREEN can shift execution toward lower-carbon periods, improving its savings and outperforming \CarbonScaler at $d=36$ hours. \systemName, however, remains more effective, reducing emissions by 15.1\% compared to \GREEN.
Lastly, the figure shows that, across baselines, increasing the delay increases carbon savings but with diminishing returns ~\cite{Hanafy2023:CarbonScaler, Sukprasert2024:Limitations}.

\begin{figure}[t]
    \centering
    \includegraphics[trim={1cm 7.2cm 0 0},clip, width=0.35\textwidth]{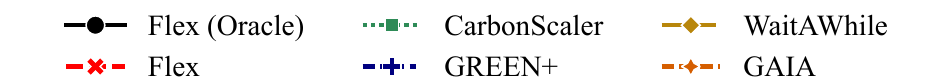}\\
    \hfill
    \begin{subfigure}[b]{0.49\linewidth}
    \includegraphics[width=\textwidth]{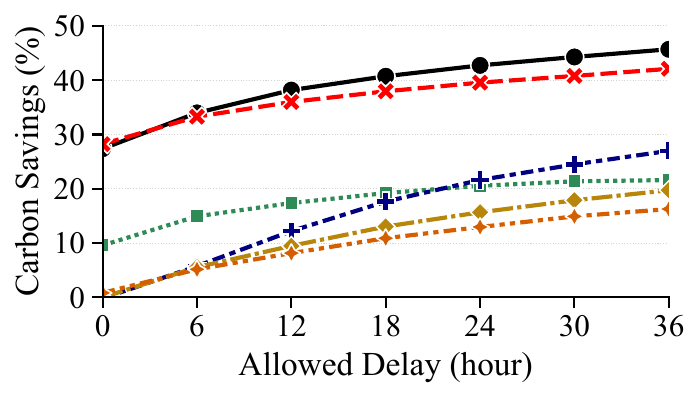}
    \caption{Carbon Savings (\%)}
    \label{fig:delay_carbon}
    \end{subfigure}
    \hfill
    \begin{subfigure}[b]{0.49\linewidth}
    \includegraphics[width=\textwidth]{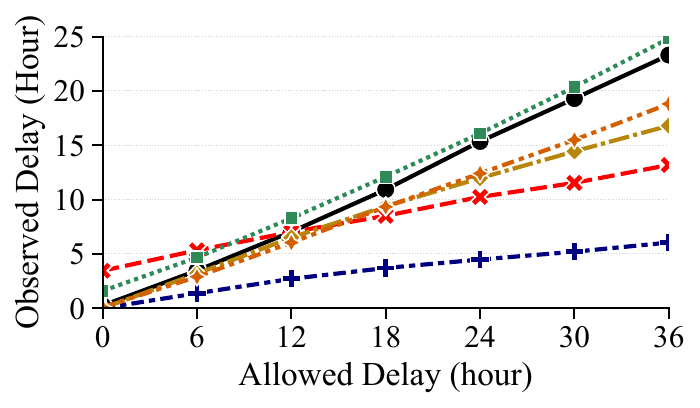}
    \caption{Observed Delay}
    \label{fig:delay_waiting}
    \end{subfigure}
    \hfill
    \hfill
    \caption{Impact of the allowed delay (slack) on the carbon savings (a) and observed delay (b).}
    \Description{Impact of the allowed delay (slack) on the carbon savings (a) and observed delay (b).}
    \label{fig:delay}
\end{figure}

~\autoref{fig:delay_waiting} shows the average observed delay across baselines. As expected, the observed delay increases with the allowed delay. The figure shows that when the allowed delay $d=0$, \systemName and \CarbonScaler incur higher observed delays of 3.4 and 1.6 hours, which explains the additional carbon savings they achieve at $d=0$ in ~\autoref{fig:delay_carbon}. Moreover, although not visible, some jobs in the oracle also exceed their allowed delay (i.e., a non-feasible schedule), which we fix by extending their delays. However, as the allowed delay increases, \systemName requires less delay, as it greedily schedules resources at the earliest possible moment when resources are available. Lastly, we note that \GREEN exhibits a lower observed delay time because it starts jobs whenever the current carbon intensity is below the mean of the remaining delay window, leading it to use the available slack less aggressively than policies that explicitly optimize for future low-cost periods.

\noindent\emph{\textbf{Key Takeaways.} \systemName can operate across different settings, outperforming the runner-up baseline by up to 15.1\% and coming within 3.6\% of the \systemOracle.}

\subsection{Effect of Workload Characteristics}\label{sec:eval_ccs}
The characteristics of the workload traces (\eg arrival rates or job scalability) affect the benefits of the scheduling and resource provisioning decisions. 
Using the carbon intensity signal, we assess the impact of workload elasticity and workload traces. 
\begin{figure}[t]
    \centering
    \includegraphics[width=\linewidth]{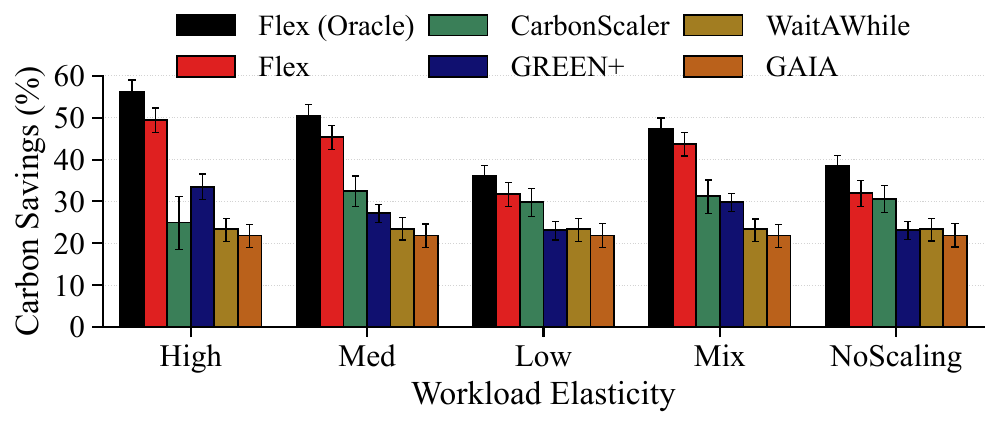}
    \caption{Impact of workload elasticity on carbon savings.}
    \Description{Impact of workload elasticity on carbon savings.}
    \label{fig:elasticity_multi}
\end{figure}

\noindent\emph{\textbf{Effect of Job Elasticity.}}
The elasticity of workloads is crucial for achieving significant carbon savings, as it enables workloads to take advantage of periods of low carbon intensity.
~\autoref{fig:elasticity_multi} illustrates the impact of workload elasticity, comparing carbon savings across workloads with varying characteristics. We use three fixed-scaling scenarios, $N$-body($N\mathord{=}100k$), $N$-body($N\mathord{=}2k$), and Jacobi($N\mathord{=}1k$), denoted high, moderate, and low elasticity (see ~\autoref{tab:workloads}), our primary ``Mix'' scenario that randomly assigns profiles, and a ``NoScaling'' scenario where jobs can only be paused, highlighting the benefits of \systemName's learning-based resource provisioning approach.
As demonstrated, workloads with enhanced scaling can achieve greater carbon savings, reducing carbon emissions by up to 56.1\% and 49.5\% for the highly scalable workloads under \systemOracle and \systemName, respectively. 
In addition, aside from the scaling profile, \systemName's performance is similar to \systemOracle's, coming within 3.4\%-6.6\% of \systemOracle. Moreover, the figure shows \GREEN across workload elasticity scenarios, highlighting how it can leverage elasticity when the workload is highly elastic. Nonetheless, the results demonstrate that \systemName achieves 16\% higher carbon savings.

Importantly, the figure shows the benefits of \systemName's data-driven approach. Even without scaling, \systemName can achieve higher carbon savings than baselines, achieving 1.4\% more savings than \CarbonScaler, which uses suspend-resume scheduling.
Moreover, the figure illustrates how different approaches perform under different elasticity profiles.
For instance, \CarbonScaler cannot take advantage of high elasticity and performs worse than other baselines. The reason is that all workloads use similar schedules, causing them to run during periods of higher carbon and to default to the lowest-scale run-to-completion behavior. 
In contrast, \GREEN performs best when workloads are highly elastic but behaves similarly to non-elastic baselines when they are not.
Lastly, baselines that do not use scaling (\eg \WaitAWhile) behave consistently across workloads.

\noindent\emph{\textbf{Workload Traces.}}
The characteristics of the trace (\eg average job length) dictate the potential carbon savings and the benefits of elastic scaling. 
\autoref{fig:traces_effect} illustrates the carbon savings across the \emph{Azure} trace~\cite{azure-data-paper}, \emph{Alibaba} trace~\cite{weng2022mlaas}, and \emph{SURF} trace~\cite{Chu2024:SURF}, where a maximum cluster capacity is selected to achieve 50\% utilization. As shown, \systemName can attain significant carbon savings across traces, ranging from 43.7\% (3.6\% from \systemOracle) for the \emph{Azure} trace to 36.9\% (5.7\% from \systemOracle) for the \emph{Alibaba} trace, with the \emph{SURF} trace at 39.6\% (7\% from \systemOracle).  
The reason for these differences can be traced to variations in job length, as \emph{Azure} has a higher average job length than the other traces. This is also reflected in the disparities between elastic and non-elastic scheduling approaches, as shorter jobs do not benefit from scaling or interruptibility. This is evident in the difference in carbon savings between \systemName and \GAIA, which decreases from 22.3\% in \emph{Azure} to 11.6\% in \emph{Alibaba}, and between \WaitAWhile and \GAIA.

\noindent\emph{\textbf{Key Takeaways.} 
\systemName's data-driven approach is beneficial across workloads, elasticity, and length distributions, coming within 3.4\% to 7\% of \systemOracle, and increasing savings by up to 16\% compared to the runner-up baselines. 
}

\begin{figure}[t]
    \centering
\includegraphics[width=\linewidth]{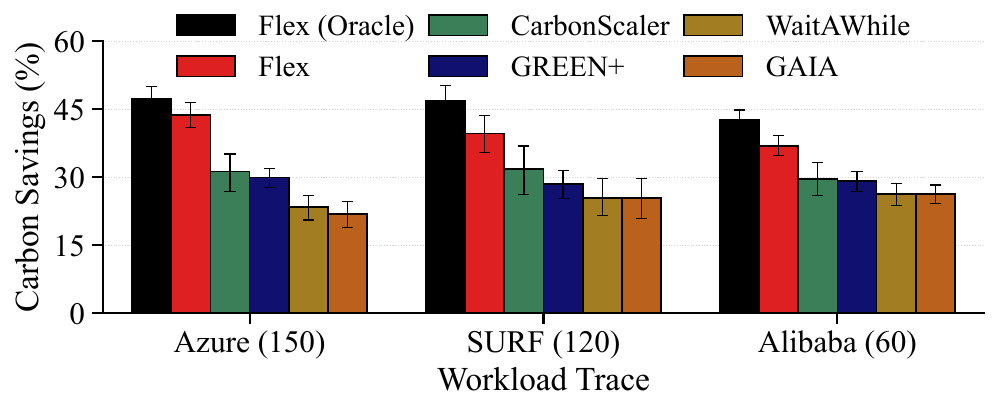}
   \caption{Carbon Savings across workload traces.%
   }
   \Description{Carbon Savings across workload traces with different job lengths and arrival patterns.}
    \label{fig:traces_effect}
\end{figure}

\begin{figure*}[t]
    \centering
    \hfil
     \begin{subfigure}[b]{0.48\linewidth}
        \centering
     \includegraphics[width=\linewidth]{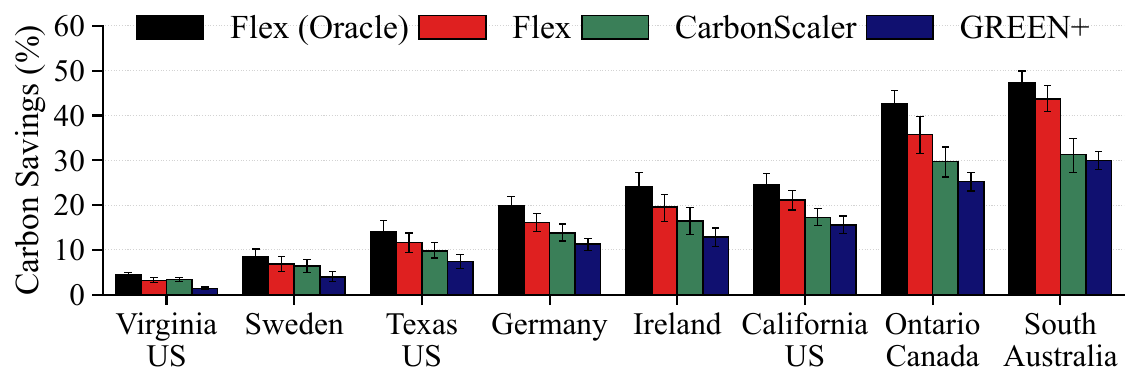}
    \caption{Carbon Intensity.}
    \label{fig:locations_multi_carbon}
    \end{subfigure}
    \hfil
    \begin{subfigure}[b]{0.48\linewidth}
        \centering
     \includegraphics[width=\linewidth]{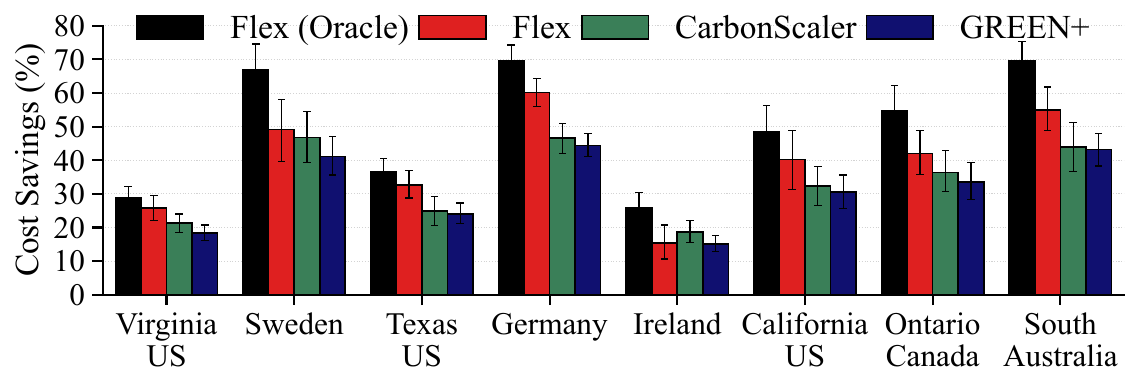}
    \caption{Energy Prices.}
    \label{fig:locations_multi_cost}
    \end{subfigure}
    \hfil
    \hfil
    \caption{Impact of signal and locations.}
    \Description{Impact of signal and locations.}
    \label{fig:locations_multi}
\end{figure*}

\begin{figure*}[t]
\hfill
\begin{minipage}[t]{.38\linewidth}
\vspace{0pt}
      \centering
   \includegraphics[width=\linewidth]{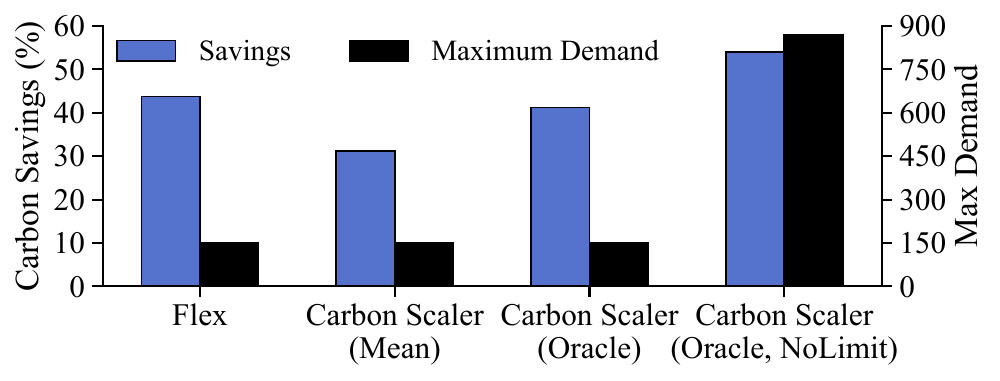}
    \caption{Carbon Savings under different assumptions and capacity bounds.}
    \label{fig:cs_error}
\end{minipage}
\hfill
\begin{minipage}[t]{.19\linewidth}
\vspace{0pt}
  \centering
    \includegraphics[width=\linewidth]{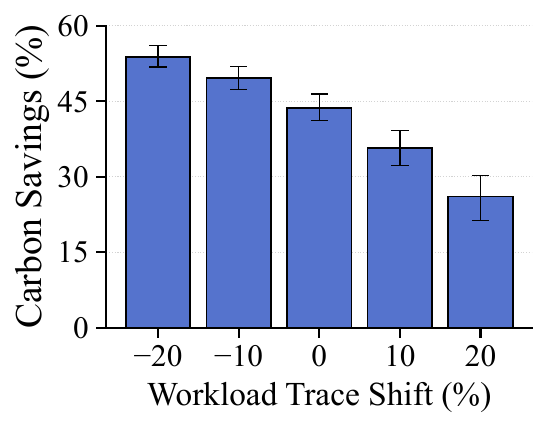}
        \caption{Impact of \\distribution shifts.}
        \Description{Impact of distribution shifts.}
        \label{fig:trace_shift}%
\end{minipage}
\hfill
\begin{minipage}[t]{.41\linewidth}
\vspace{0pt}
    \centering
    \includegraphics[width=\linewidth]{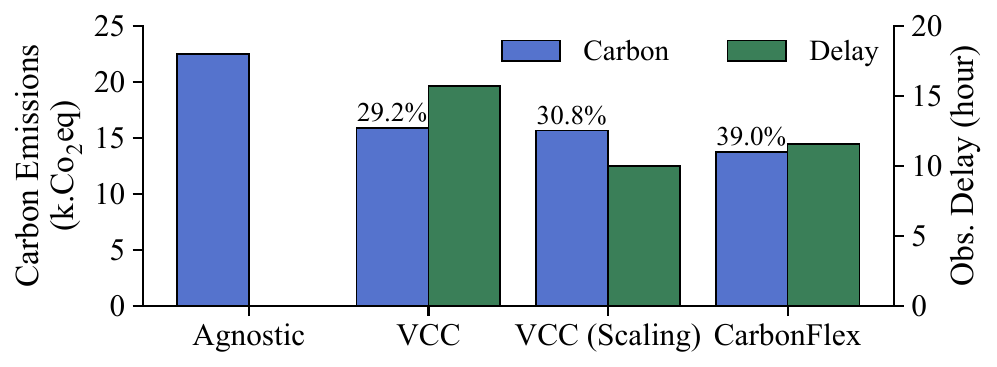}
    \caption{Comparing \systemName with carbon-aware capacity provisioning.}
    \label{fig:vcc_compare}
\end{minipage}
\end{figure*}

\subsection{Effect of Cloud Location and Signal}\label{sec:eval_location_signal}

As noted in ~\autoref{sec:background_CI}, the supply mix and diurnal changes in demand significantly affect the carbon intensity and energy prices across locations.
\autoref{fig:locations_multi} shows the impact of location and signal choice across eight locations (see \autoref{fig:carbon_trace_ccs}), reporting carbon savings under the carbon-intensity signal (\autoref{fig:locations_multi_carbon}) and cost savings under the energy price (\autoref{fig:locations_multi_cost}).
\autoref{fig:locations_multi_carbon} shows that \systemName closely matches \systemOracle's carbon savings, coming within 0.9\% to 6.31\% of \systemOracle's savings. 
As shown, carbon savings track the signal's variability, as highly variable locations (see \autoref{fig:carbon_trace_ccs}) save more than places such as Virginia, US, where 85\% of 2022 electricity came from non-variable sources~\cite{eia_virginia}, limiting saving opportunities. 
Importantly, the figure also shows that as the variability increases, the difference between \systemName and other baselines (\eg \CarbonScaler and \GREEN) also increases, reaching a gap of 12.5\%, highlighting the benefits of \systemName in reducing emissions.

\autoref{fig:locations_multi_cost} shows that \systemName outperforms other baselines, reducing electricity cost by up to 60\%, increasing the savings by 13.6\% compared to the closest runner-up approach.
Moreover, the figure highlights that these trends generalize when using the more variable cost signal, where \systemName closely matches the \systemOracle (with a mean difference of 3.44\%).
Nonetheless, it is worth noting that \systemName tends to trail \systemOracle more under the cost signal, whose values span multiple orders of magnitude, making the oracle's decisions harder to match.

\noindent\emph{\textbf{Key Takeaways.} \systemName outperforms baselines on both carbon intensity and electricity prices, with savings that scale with the signal's variability, achieving up to 43.7\% carbon savings and 60\% cost savings, outperforming the runner-up by 13.6\%.}

\subsection{Impact of Job Knowledge and Capacity}
\autoref{fig:cs_error} shows the carbon savings and peak resource requirements of \systemName, which needs no per-job information, and \CarbonScaler, which relies on historical job length, compared against \CarbonScalerOracle, which knows exact job length. We also add \CarbonScalerMax, which knows job length and allows unlimited servers, an upper bound on achievable savings.
As shown, \systemName performs within 2\% of \CarbonScalerOracle, accurately capturing the workload distribution without per-job knowledge. \systemName also reaches 81.1\% of the maximum achievable savings (\CarbonScalerMax). However, we note that these savings follow an aggressive scheduling pattern, stacking all workloads in low-carbon periods and increasing the peak capacity by $\sim$5.8$\times$.

\subsection{Effect of Workload Distribution Shifts}
Another key assumption in \systemName is that historical and real-time workload traces are similar. \systemName copes with departures from this assumption by relearning its parameters over a rolling window. \autoref{fig:trace_shift} tests this by introducing sudden distribution shifts, varying the inter-arrival rate and job length from -20\% to 20\% (zero is the original trace), which changes cluster utilization. As shown, \systemName adapts to each shifted distribution and continues to capture the savings available at that utilization, from 53.9\% when lower arrival rates and shorter jobs free capacity to 26\% when higher arrival rates raise utilization.

\subsection{Learned vs.\ Forecast-Based Provisioning}
Prior signal-aware provisioning uses explicit forecasts of the grid signal and compute demand~\cite{Radovanovic2023:VCCPaper,Lin2023:Adapting} and does not model workload or scheduling parameters.
\autoref{fig:vcc_compare} compares \systemName against two variants of Google's VCC~\cite{Radovanovic2023:VCCPaper}. \emph{VCC} follows the VCC capacity curve and schedules FCFS, whereas \emph{VCC (Scaling)} pairs it with \systemName's elastic scheduling. All jobs use a 24-hour delay budget.
As shown, adding our elastic scheduling to VCC's forecast-based provisioning (\emph{VCC (Scaling)}) reduces carbon emissions by 1.6\% and waiting time by 36\%. More importantly, \systemName's learned provisioning outperforms this forecast-based provisioning by a further 8.2\%, showing that learning provisioning beats forecasting it.

\subsection{System Overheads}\label{sec:eval_overhead}
Lastly, we used our prototype to quantify \systemName's overheads. Running the offline oracle for a week-long trace took 2 to 10 minutes, and matching the current state to the closest oracle states took 1 to 2 ms.
Workload profiling took 30 seconds per CPU scale (16 total) and 1 minute per GPU scale (8 total), about 2 hours overall. Checkpoint/Restore overhead depends on the application's memory footprint~\cite{flint}.
Our largest, ViT-B/32 (see \autoref{tab:workloads}), took 2 and 0.3 seconds to checkpoint and restore. Finally, provisioning EC2 instances takes 3 minutes for \texttt{c8g.medium} CPU and 5 minutes for \texttt{g6.xlarge} GPU instances.

%% file: sections/6-related.tex
\noindent\emph{\textbf{Signal-aware scheduling and provisioning.}}
Prior carbon- and price-aware schedulers temporally shift work to favorable hours, for individual jobs~\cite{Hanafy2023:CarbonScaler, Wiesner2021:WaitAwhile, Lechowicz2023:DTPR, Souza2023:Ecovisor} or across a fixed cluster~\cite{Hanafy2024:GoingGreen, Xu2025:Green}. A complementary line adjusts \emph{capacity}, growing and shrinking the server pool with the grid signal~\cite{Radovanovic2023:VCCPaper, Zhang2021:VariableCapacityChallanges, Lin2023:Adapting, Wijayawardana2026:VariableCapacityVMs} or participating in demand response~\cite{Xing2023:CarbonResponder, Lin2025:Middlebox, Zhang2019:DRParticipation,Acun2026:Flexibility}. These treat scheduling and provisioning separately or plan against demand forecasts. Closest to our work, SPADE~\cite{Lechowicz2026:SPADE} is signal-aware and provisions dynamically, but targets DAG data-processing with threshold-based deferral. \systemName instead \emph{co-designs} capacity provisioning with elastic job scheduling, resolving their circular dependence by replaying oracle decisions rather than forecasting demand.

\noindent\emph{\textbf{Elastic scheduling.}}
Elastic and malleable schedulers~\cite{Tarraf2024:Malleability} scale jobs to optimize makespan~\cite{Peng2018:Optimus, Amico2019:Slowdown}, throughput~\cite{Subramanya2023:Sia, Xiao2020:AntMan, Qiao2021:Pollux}, energy~\cite{You2023:Zeus}, and carbon~\cite{Xu2025:Green,Hanafy2023:CarbonScaler}. \systemName differs by making the cluster grid-responsive and jointly sizing its capacity.

\noindent\emph{\textbf{Data-Driven Resource Management.}}
\systemName's data-driven policy relates to case-based reasoning~\cite{WM1994:CBRBook} and hindsight learning~\cite{Sinclair2023:Hindsight}, which reuse solutions from similar past cases.
Related data-driven resource management approaches match incoming jobs to profiled ones (Paragon~\cite{Delimitrou2013:Paragon}), infer SLOs from past runs (Morpheus~\cite{Jyothi2016:Morpheus}), or predict workloads from telemetry (Resource Central~\cite{azure-data-paper}).
\systemName differs by distilling the oracle into simple run-time parameters, not a case store or learned policy.

%% file: sections/7-conclusion.tex
This paper presents \systemName, a grid-responsive resource manager for cloud clusters. \systemName employs a data-driven approach to guide near-optimal scheduling and provisioning decisions for elastic batch jobs across CPU and GPU clusters. Our evaluation showed that \systemName reduces energy-related emissions by up to 57\%, is within 2.1\% of an oracle, and generalizes to a real-time electricity-price signal, cutting costs by up to 60\%. In the future, we plan to extend our grid-responsive provisioning and scheduling approaches to interactive workloads and heterogeneous cluster settings.

%% file: sections/9-theory-appendix.tex
\section{Optimality of \systemName Oracle}\label{app:proof}

We prove \autoref{theorem:CFO} for elastic jobs with $k_j^{\min}=1$, the setting of our evaluation (\autoref{app:elastic_workloads}). Non-elastic jobs ($k_j^{\min}=k_j^{\max}$) are handled by the same algorithm as a single indivisible allocation. Throughout, we assume a non-negative grid signal ($c_t\ge0$), homogeneous resources, and negligible scaling overhead, and that the instance admits a feasible schedule.

We represent any schedule as a multiset of unit \emph{moves}. Placing the $k$-th resource unit of job $j$ in slot $t$ costs $c_t$ and contributes benefit $\Delta_j(k)=\Phi_j(k)-\Phi_j(k-1)$ toward job $j$'s work requirement $l_j$. \autoref{alg:offline} ranks moves by their \emph{cost-per-benefit ratio}
\begin{equation}
  \rho_{j,t}^{(k)} \;=\; \frac{c_t}{\Delta_j(k)},
  \label{eq:ratio0}
\end{equation}
for $t\in[a_j,e_j]$ and $k\in[1,k_j^{\max}]$, breaking ties by the earliest deadline $e_j$, and greedily applies the feasible move of minimum ratio until every job meets its requirement.

\begin{lemma}[Monotonicity of local ratios]\label{lem:mono}
For a fixed $(j,t)$, the ratio $\rho_{j,t}^{(k)}$ is non-decreasing in $k\ge 1$.
\end{lemma}
\begin{proof}
Since $\Delta_j(k)$ is non-increasing in $k$ (concavity) and $c_t\ge0$, $k\mapsto c_t/\Delta_j(k)=\rho_{j,t}^{(k)}$ is non-decreasing.
\end{proof}

\begin{proof}[Proof of \autoref{theorem:CFO}]
We will prove \autoref{theorem:CFO} by contradiction.
Let $\mathsf{G}$ denote the greedy schedule produced by \autoref{alg:offline}, and suppose there exists a feasible
schedule $\mathsf{OPT}$ with strictly smaller cost.
Represent any schedule by the multiset of \emph{moves}
it applies (each activation counts as $k_j^{\min}$ resource units in slot $t$ and supplies benefit $\Phi_j(k_j^{\min})$ to job $j$, and each incremental move adds $1$ unit of capacity in slot $t$ and supplies benefit $\Delta_j(\cdot)$ to job $j$). List greedy moves in order of selection:
\[
\mathcal{M}_{\mathsf{G}}=\bigl(m_1,m_2,\dots,m_U\bigr),
\]
with corresponding non-decreasing cost-per-benefit ratios
$\rho(m_1)\le\rho(m_2)\le\cdots\le\rho(m_U)$ (Lemma~\ref{lem:mono}).

Let $\ell$ be the first index where $\mathsf{OPT}$ and $\mathsf{G}$ diverge, i.e., after the first $\ell{-}1$ moves they coincide as multisets.
Therefore, the \emph{residual instance} (remaining job work and residual time capacities) is identical for both at step $\ell$.
Let $m_\ell$ be greedy’s next move, with ratio $\rho(m_\ell)$, and let $m_b$ with ratio $\rho(m_b)$ be any move used by $\mathsf{OPT}$ at this step.
Because $m_b$ is feasible for $\mathsf{OPT}$ and residuals match, $m_b$ is also feasible for greedy at step $\ell$.
By definition of \autoref{alg:offline}, greedy picks a feasible move of \emph{minimum} ratio, hence
\begin{equation}
  \rho(m_\ell)\ \le\ \rho(m_b).
  \label{eq:minratio}
\end{equation}

To insert $m_\ell$ into $\mathsf{OPT}$ while keeping it feasible, we rebalance along an alternating trail in the bipartite job-slot graph, adding $m_\ell$ and alternately removing and adding units at the visited jobs and slots so that every job's total benefit (its work requirement) and every slot's total usage (its capacity $M$) are preserved. Every unit removed along the trail is feasible in the residual instance at step $\ell$, so by greedy minimality~\eqref{eq:minratio} and \autoref{lem:mono} it has ratio at least $\rho(m_\ell)$. Because moves are ordered by non-decreasing cost-per-benefit, this exchange does not increase the total cost~\cite{greedy_optimal}.

Repeating the exchange transforms $\mathsf{OPT}$ into $\mathsf{G}$ without ever increasing cost, so $\mathrm{cost}(\mathsf{G})\le\mathrm{cost}(\mathsf{OPT})$, contradicting $\mathrm{cost}(\mathsf{OPT})<\mathrm{cost}(\mathsf{G})$. Therefore \autoref{alg:offline} returns a minimum-cost schedule.
\end{proof}

\section{Elastic Workloads}\label{app:elastic_workloads}
\autoref{tab:workloads} describes our CPU and GPU workloads implemented using MPI~\cite{mpi} and PyTorch~\cite{pytorch}, respectively. The table presents the workload names, communication sizes in MB, and scalability, categorizing applications as High, Moderate, or Low. 
We obtain profiles through one-time profiling that iterates over possible nodes between $[k^{min}, k^{max}]$ and runs for a brief duration (typically a few minutes).
In our current experiments, we profiled workloads on AWS at various scales. CPU loads were profiled between [$k^{min}=1, k^{max}=16$] CPU cores, while GPU loads were profiled from [$k^{min}=1, k^{max}=8$] due to limitations in GPU capacity.
However, we note that \systemName can also support alternative workload profiling approaches from previous work~\cite{paleo, Oyama2016, Pei2019, justus2018, cai2017neuralpower, Peng2018:Optimus, performance_modeling}.
\input{tables/workloads}

%% file: tables/workloads.tex
\begin{table}[ht]
\caption{Details of the elastic workloads in evaluation.}
\label{tab:workloads}
%\footnotesize
\resizebox{\linewidth}{!}{%
\begin{tabular}{l|c|c|c}
\toprule
\textbf{Workload} & \textbf{Impl.} & \textbf{Comm. Size} & \textbf{Scalability}\\
\midrule
$N$-body($N\mathord{=}
100k$) ~\cite{nbodysimulation} & MPI & 5.3 MB$^*$ & High\\
$N$-body($N\mathord{=}10k$) ~\cite{nbodysimulation} & MPI & 0.53 MB$^*$ & High\\
$N$-body($N\mathord{=}3k$) ~\cite{nbodysimulation} & MPI &  0.16 MB$^*$ & Moderate\\
$N$-body($N\mathord{=}2k$) ~\cite{nbodysimulation} & MPI & 0.1 MB$^*$ & Moderate\\
Jacobi($N\mathord{=}3k$) \cite{verschelde2016parallel} & MPI & 51.2 MB$^*$ & Low\\
Jacobi($N\mathord{=}2k$) \cite{verschelde2016parallel} & MPI & 28.6 MB$^*$ & Low\\
Jacobi($N\mathord{=}1k$) \cite{verschelde2016parallel} & MPI & 7.16 MB & Low*\\
AlexNet \cite{AlexNet} & Pytorch &   233.1 MB& Low\\
ResNet18 \cite{he2016resnet} & Pytorch&   44.7 MB & Low\\
ResNet50\cite{he2016resnet} & Pytorch &    97.8 MB & Moderate\\
ResNet101\cite{he2016resnet} & Pytorch &    170.5 MB & High\\
EffNet-S~\cite{tan2021efficientnetv2} & Pytorch &   82.7 MB & High\\
ViT-B/32\cite{vit} & Pytorch &   336.6 MB & Moderate\\
\bottomrule
\multicolumn{4}{l}{$^*$ This application present our least scalable workload (See ~\autoref{fig:elasticity}).}\\
\end{tabular}%
}
\end{table}